\documentclass[reprint,amsmath,amssymb,aps,pra,superscriptaddress]{revtex4-1}
\usepackage{hyperref}
\usepackage{graphicx}
\usepackage{dcolumn}
\usepackage{bm}
\usepackage{color}
\usepackage{xcolor}
\usepackage{bbold}
\usepackage{soul}
\usepackage{qcircuit}
\usepackage{epsfig,epic}
\usepackage{epstopdf}
\usepackage{amsthm}
\usepackage{subcaption}
\usepackage{enumerate}
\usepackage{cleveref}
\newcommand{\newc}{\newcommand}
\newc{\beq}{\begin{equation}}
\newc{\eeq}{\end{equation}}
\newc{\kt}{\rangle}
\newc{\br}{\langle}
\newc{\red}{\textcolor{red}}
\newc{\beqa}{\begin{eqnarray}}
\newc{\eeqa}{\end{eqnarray}}
\newc{\tr}{\mbox{tr}}
\newc{\ovl}{\overline}
\usepackage[flushleft]{threeparttable} 
\usepackage{tabularx}
\usepackage{array}
\usepackage{soul}

\newtheorem{theorem}{Theorem}[section]
\newtheorem{corollary}{Corollary}[theorem]
\newtheorem{proposition}{Proposition}
\newtheorem{lemma}[theorem]{Lemma}
\newc{\longra}{\longrightarrow}

\newc{\blue}{\textcolor{blue}}

\hypersetup{
    colorlinks=true,
    linkcolor=blue,
    filecolor=magenta,      
    urlcolor=cyan,
}
\makeatletter
\let\Hy@backout\@gobble
\makeatother

\begin{document}

\title{Entanglement measures of bipartite quantum gates and their thermalization\\
 under arbitrary interaction strength}
\author{Bhargavi Jonnadula}
\address{School of Mathematics, University of Bristol, BS8 1UG, United
Kingdom}
\author{Prabha Mandayam} 
\affiliation{Department of Physics, Indian Institute of Technology Madras, Chennai, India~600036}
\author{Karol {\.Z}yczkowski}
\affiliation{Smoluchowski Institute of Physics, Jagiellonian University, Cracow, Poland}
\address{Center for Theoretical Physics, Polish Academy of Sciences, Warsaw, Poland}
\author{Arul Lakshminarayan}
\affiliation{Department of Physics, Indian Institute of Technology Madras, Chennai, India~600036}
\affiliation{Max-Planck-Institut f\"{u}r Physik komplexer Systeme, N\"{o}thnitzer Strasse 38, 01187 Dresden, Germany}


\date{October 4, 2020}

\begin{abstract}

Entanglement properties of bipartite unitary operators are studied via their local invariants, namely the entangling power $e_p$ and a complementary quantity, the gate typicality $g_t$. 
We characterize the boundaries of the set $K_2$ representing all two-qubit gates projected onto the plane $(e_p, g_t)$
showing  that the fractional powers of the \textsc{swap} operator form a parabolic boundary of $K_2$, 
while the other bounds are formed by two straight lines. 
In this way a family of gates with extreme properties is identified and analyzed. We also show that the parabolic curve representing powers of \textsc{swap} persists in the set $K_N$, for gates of higher dimensions ($N>2$). Furthermore, we study entanglement of bipartite quantum gates applied sequentially $n$ times and analyze the influence of interlacing local unitary operations, which model generic Hamiltonian dynamics. 
An explicit formula for the entangling power a gate applied $n$ times averaged over random local unitary dynamics
is derived for an arbitrary dimension of each subsystem.
This quantity shows an exponential saturation to the value predicted by the
random matrix theory (RMT), indicating ``thermalization" in the
entanglement properties of sequentially applied quantum gates that can have arbitrarily small, but
nonzero, entanglement to begin with.
The thermalization is further characterized by the 
spectral properties of the reshuffled and partially transposed unitary matrices.
  

\end{abstract}

\maketitle

\section{Introduction}

A clutch of quantities such as state entanglement, operator entanglement, operator
scrambling, out-of-time-ordered correlators, and various measures of mutual information
are being currently actively pursued as a means to understand information transport in
complex quantum systems and to characterize quantum chaos
\cite{Swingle2016,Nahum2017,Helu2017,Luitz2017,Akshay2018,Hosur2016,Chalker2018,Keyserlingk2018}.
Entangling power of the time evolution operator has
been studied since its introduction as a state independent measure
\cite{Zanardi2000,Zanardi2001,Wang2002}, and is the average entanglement an operator
produces when acting on product, unentangled, states. On the other hand, operator entanglement quantifies to what extent a given operator,
 treated as a vector in the Hilbert-Schmidt space of operators,
  is close to the tensor product \cite{Zyczkowski2004}. Operator entanglement and
entangling power have both recently been applied to many-body systems, in particular
in the context of spin-chains and conformal field theories \cite{Luitz2017,Dubail17,
PalLak2018}, where it has been found useful to distinguish between integrable and
non-integrable systems as well as in the analysis of the many-body-localization
transition.


The study of the operator entanglement and entangling power of the time evolution operator $\exp( -i Ht/\hbar)$, or its time-ordered version
if the Hamiltonian $H$ is a function of time, is of fundamental import in the growth of subsystem entropy and complexity of 
closed systems ranging from two large bipartite systems to
quantum spins on lattices \cite{Lubkin1993, MillerSarkar1999,BandyoLak2002,Fujisaki2003, BandyoLak2004,Dobrzanski2004,Linden2009,Chaudhury2009,Neill16, Lakshminarayan2016,Schuch2008,Abreu2007, Calaberese2018, Bertini2019,Bravyi2007,PetitJean2006,Madhok2008, Jethin2020}. The action of this time-evolution operator on unentangled states generally creates multipartite entanglement. From another perspective, the Heisenberg evolution results in operator entanglement and scrambling in the space of operators~\cite{Nahum2017,Chalker2018,Keyserlingk2018, Moudgalya2019}. A central aspect of this paper is the study of dynamics of quantum entanglement in products of unitary matrices, which are interpreted as  time-evolution operators,  with the number of terms in the product playing the role of discrete time.

From the point of view of quantum computing \cite{NielsenChuang}, gate operations ordered in time are the source of information transfer. 
Products of unitary operators are therefore natural objects to study as they form building blocks for quantum algorithms.
Random quantum circuits with random unitary operators providing interaction among qubits have been studied in this context \cite{Emerson2003,Harrow2009, KZ13}. They are known to be approximate unitary $t-$designs that simulate Haar distributed unitaries \cite{DCEL09,BHM16}. 
Models of random quantum circuits have been studied in many other contexts including 
randomised benchmarking \cite{EAZ05},  entanglement spreading,
 scrambling and many-body localization \cite{Chalker2018,Keyserlingk2018,HuseCirac2018}.


Random quantum circuits are constructed by arbitrarily choosing pairs of quNits between which the interactions are described by random unitary matrices \cite{CNZ10}, typically Haar distributed. In a departure from this standard formalism, we are primarily interested in the role of {\it local} random unitaries with a tensor product structure, which interlace sequential dynamics described by a fixed non-local gate $V$ acting on a bipartite system. 
Such a bipartite structure could form a building block for more general random quantum circuits with fixed, possibly atypical, nonlocal gates and random or generic local interaction.

Another setting in which products of  nonlocal unitaries interspersed with local ones arise naturally are in kicked systems which are being extensively used. In this context, the object of interest could be powers of the Floquet operator $U$ \cite{BandyoLak2002,BandyoLak2004,Dobrzanski2004,PalLak2018,Luitz2017}, or if the local 
Hamiltonians are non-autonomous, products of propagators across consecutive periods of the kicking. In particular, let
\beq
\label{eq:kickhamil}
H=H_A(t)\otimes \mathbb{1}_B + \mathbb{1}_A \otimes H_B(t) +  H_{AB}
\sum_{n=-\infty}^{\infty} \delta(t/\tau -n),
\eeq
be the Hamiltonian and
\beq
u_{A j}=\mathcal{T}e^{-i \int_{(j-1)\tau}^{j\tau} H_{A}(t) dt},
\eeq  
where $\mathcal{T}$ denotes time ordering.
The time evolution operator between (just before) kicks $j-1$ and $j$ is 
\beq
\mathcal{U}_j= \left( u_{A\, j} \otimes u_{B\, j}\right) \, U
\eeq
where $U=e^{-i \tau H_{AB}}$ arises from the non-local interaction at time $j\tau$ and $u_{Bj}$ is defined similar to $u_{Aj}$.
This is the Floquet operator across the time period $\tau$ between the $j-1$ and $j^{th}$ kicks.
The propagator across $n$ kicks is 
\beq
\label{eq:Upower(n)}
\mathcal{U}^{(n)}=\mathcal{T} \prod_{j=1}^n \mathcal{U}_j= \mathcal{T} \prod_{j=1}^n(u_{A_{j}}\otimes u_{B_{j}})\, U.
\eeq
The brackets around the ``power" $n$ is to indicate that there are $n$ different terms in the product
generally and the time ordering will be assumed below and hence not explicitly indicated. Most systems that have been studied are such the $u_{Aj}$ and $u_{Bj}$ are independent 
of $j$, that is the local Hamiltonians are autonomous. This leads to a time-periodic
system with $\mathcal{U}$ as the Floquet operator and we are then interested in the powers $\mathcal{U}^n$. 
However, a product of unitary operators with different local operators occur in contexts such as time-dependent quenches, 
for example see the study \cite{Mishra_2014} for quenched kicked Ising spin chains.

Past work has shown that while bipartite local unitaries $u_A\otimes u_B$ have no entangling power, layering or
interspersing them in time with entangling gates provides a crucial role for random local unitaries \cite{JMZL2017,Mandarino2018}. 
Local unitary gates are easier to apply in an experiment
and are thus  naturally ``cheaper" than nonlocal entangling gates. However, the role of such local unitary gates in creating Haar random unitaries or in achieving thermalization is to our knowledge not sufficiently explored. Specifically, we focus here on the thermalization of 
the entangling power of the unitary operator $\mathcal{U}^{(n)}$, as defined in Eq.~(\ref{eq:Upower(n)}),
 where thermalization in understood to mean that  after a certain number $n$
of interaction times the quantities studied reach the typical values corresponding to the Haar average over the unitary group.
We find that the natural quantities to study are indeed the entangling power, $e_p$, and a complementary quantity defined in~\cite{JMZL2017}, as the ``gate typicality", $g_t$. In particular, we are interested in the entangling power and gate typicality of $\mathcal{U}^{(n)}$. The importance of the operator entanglement and entangling power stems from
their invariance under local unitary operations and hence measure the essential nonlocal content of the process. 

This paper contains two complementary but in some ways distinct motivations and results, which for the sake 
of the convenience of the reader we enumerate below.
\begin{enumerate}[(i)]
\item 
The first part of this paper is dedicated to visualizing the entanglement landscape of bipartite gates in dimensions $N^2$, in terms of entangling power and gate typicality. For the case of qubits, $N=2$, the picture is complete and we show in detail the various gates that make up the ``phase-space" spanned by these two local invariants. We prove the existence of a boundary consisting exclusively of the fractional powers of the \textsc{swap} operator. Of special interest are gates that maximize entangling power, in the sense that they attain bounds set by the dimensionality $N$. It is known that for two-qubit systems such a gate does not exist \cite{Higuchi2000}, while  the maximal entangling power is attained for {\sc cnot} and related gates \cite{Zanardi2000,JMZL2017}. 

\item If the first part is about the ``kinematics" of the entangling power and gate typicality, the second is a study of its ``dynamics", via
the entangling power of the products of unitaries. We generalize earlier results \cite{JMZL2017} for equal subsystem dimensionality to the important case when the two subsystems could be of different dimensions. In a central result in this context, we demonstrate the exponentially fast thermalization of the average entangling power of $\mathcal{U}^{(n)}=\Pi_{j=1}^n \mathcal{U}_j$ with time, to that of a typical unitary operator. Furthermore, we show that there are signatures of such a thermalization in the spectra of the opeators obtained by reshuffing and partial transposition (both permutations) of the time-evolution operator.
\end{enumerate}

Thus the second part of this work shows the thermalization of the entangling power of $U$ under time evolution with non-autonomous local evolutions. Such exponential saturation also seems to provide excellent approximations in the case of autonomous Floquet systems \cite{JMZL2017} whose dimensions are not very small, although the circumstances under which this holds needs further investigation. Thus we expect applications not only for coupled chaotic systems such as the kicked top and the kicked rotor, but also to many-body systems such as the kicked and tilted field Ising models \cite{Luitz2017,PalLak2018}.
The exponential approach of the average entangling power of $\mathcal{U}^{(n)}$ to the Haar average is determined solely by $e_p(U)$, the entangling power of the interaction. This demonstrates that any nonzero value of the entangling power, however small, is sufficient to thermalize its powers interspersed with random local operators.

Apart from the entangling power, the gate typicality \cite{JMZL2017} also has a simple exponential approach to the global RMT average,
depending solely upon the $g_t(U)$, the gate typicality of the interaction. In fact, this formed an important basis for the introduction of this quantity that is naturally singled out. In contrast, the thermalization of other local invariants, such as the operator entanglement, are sums of exponentials with different rates. The extremal values of the gate-typicality, $g_t=0$ and $g_t=1$, correspond to local gates and the \textsc{swap} operator, respectively, while the average value, $g_t=1/2$ (for equal subsystem dimensions), characterizes the Haar average over the entire set of bipartite unitary gates. Thus the entangling power and gate typicality are local invariants associated with the interaction $U$ that determine the complexity of products such as $\Pi_{j=1}^n \mathcal{U}_j$.

The paper is organized as follows. Section~(\ref{sec:localinv}) will introduce in detail all the relevant quantities, including operator entanglement and entangling power. Section~(\ref{sec:2qubits}) studies the allowed region of the invariants for the case of two qubits. We study this via the entangling power, gate typicality ($e_p, g_t$) phase space and establish the boundaries of the allowed gates. Section~(\ref{sec:qudits}) discusses some special gates such as the Fourier and the fractional powers of \textsc{swap} in arbitrary dimensions, and give partial results for qutrits as well as conjecture that the fractional powers of \textsc{swap} form a boundary for all quNits. Finally in Section~(\ref{sec:thermal}) we study time evolution and prove the thermalization of entangling power (and gate typicality) under certain conditions. Here, we generalize our earlier result obtained in  \cite{JMZL2017}, referring to Appendix~\ref{app:details} for an elegant proof. Finally, we provide examples wherein the thermalization can be seen via approach of the partial transposed and reshuffled operators to the Girko circular law \cite{Girko1985}
and their squared singular values to the Marcenko-Pastur distribution \cite{MPlaw}. 
Section~(\ref{sec:summary}) provides a summary and outlook.

\section{Local Invariants of operators:  Entangling power and gate typicality}
\label{sec:localinv}

\subsection{Two sets of local unitary invariants and operator entanglement}

Consider a unitary operator $U$ acting on the bipartite space ${\cal H}_N^A \otimes
{\cal H}_N^B$ of two parts labeled $A$ and $B$. For simplicity we restrict attention to
spaces whose dimensions are equal (and to $N$). The generalization to unequal dimensions is treated
in Appendix \ref{app:4party}. Operators such as $U$ may be ``gates" in the language of
quantum circuits, or just quantum propagators describing evolution over some finite time.
The fact that $U$ need not be of a product form $u_A\otimes u_B$, with $u_{A,B}$ acting
on ${\cal H}_N^{A,B}$ in general implies that it is usually capable of creating
entanglement when it acts on unentangled states. Let the operator Schmidt decomposition
of $U$ be \cite{Zyczkowski2004}
\beq
\label{eq:SchmidtU}
U= \sum_{j=1}^{N^2} \sqrt{\lambda_j} \, M_{A_j} \otimes M_{B_j} ,
\eeq
where the operators on the individual spaces $M_{A_j}$ and $M_{B_j}$
are in general not unitary themselves, but form an orthonormal basis for operators on
their respective spaces,
$\tr(M_{A_j}^{\dagger} M_{A_k})=  \tr(M_{B_j}^{\dagger} M_{B_k})=\delta_{jk}$,
 where $\delta_{jk}$ is the Kronecker delta.
 The  Schmidt vector $\lambda=\{ \lambda_{j}\}_{j=1}^{N^2}$ 
 is determined by singular values of the reshuffled matrix $U^R$ -- see 
 Appendix  \ref{app:4party}.
Note that $\lambda$ is invariant under local unitary operations.  
  Unitarity of $U$ implies that 
\beq
\frac{1}{N^2}\sum_{j=1}^{N^2} \lambda_j = 1,
\eeq
so the rescaled vector of Schmidt ceofficients, $\{\lambda_{i}/N^2\}$,
can be treated as a discrete probability measure that characterizes the nonlocality of the
operator $U$.
To elaborate let 
\beq 
U \rightarrow U'= (u_{A_1} \otimes u_{B_1}) \, U  \,(u_{A_2} \otimes u_{B_2}) ,
\eeq
where $u_{A_j,B_k}$ are ``local" unitary operators. In the language of dynamics, they
constitute single-particle evolutions. The content of nonlocality of $U$ and $U'$ is
identical and hence the measures characterizing their nonlocality must be the same. In
the case of states, this constitutes the condition that all entanglement measures be
local unitary invariants. It is clear from the definition of the operator Schmidt
decomposition that the set $\{ \lambda_i\}$ are $N^2$ such invariants, as $M_{A_j}
\rightarrow u_{A_1} M_{A_j} u_{A_2}$ also constitute an operator basis consisting of
orthonormal operators, and similarly for $M_{B_j}$.

Another set of $N^2$ invariants are constructed from the operator Schmidt decomposition of
the operator product $U S$ where $S$ is the
\textsc{swap} (or flip) operator defined as
\beq
S|\phi_A\kt |\phi_B \kt = |\phi_B\kt |\phi_A\kt, \, \text{or}\, S(u_A\otimes
u_B)S=u_B\otimes u_A,
\eeq
for arbitrary states $|\phi_{A,B}\kt $ and operators $u_{A,B}$. Let 
\beq
\label{eq:SchmidtUS}
US= \sum_{j=1}^{N^2} \sqrt{\mu_j} \, \tilde{M}_{A_j} \otimes \tilde{M}_{B_j},
\eeq
be its Schmidt decomposition. As $S$ is unitary we also have that 
\beq
\frac{1}{N^2}\sum_{j=1}^N \mu_j=1.
\eeq
That the set $\{\mu_i\}$ constitute $N^2$ invariants follows from the observation that 
\beq
\begin{split}
U'S&= (u_{A_1} \otimes u_{B_1}) \, U  \, (u_{A_2} \otimes u_{B_2})S \\&=
(u_{A_1} \otimes u_{B_1}) \, U S \, (u_{B_2} \otimes u_{A_2}),
\end{split}
\eeq
and hence the Schmidt eigenvalues of $US$, the $\mu_i$, are the same
as the Schmidt eigenvalues of $U'S$. The product $SU$ does not produce any newer
invariants. 

This paper is focused on
 these two sets of invariants and quantities derived from them. 
In particular, their moments and entropies provide measures of how nonlocal the operator
$U$ is, leading to a class of \emph{operator entanglement entropies}. Here, we will be concerned with the entropies related to the second moments, given by,
\beq
E(U)=1-\frac{1}{N^4}\sum_{j=1}^{N^2} \lambda_j^2, \;\;\text{and}\;
E(US)=1-\frac{1}{N^4}\sum_{j=1}^{N^2} \mu_j^2.
\eeq
$E(U)$ and $E(US)$ are the \emph{linear} operator entanglement entropies of the operators $U$ and $US$ respectively. They take values in $[0, 1-1/N^2]$, and
$E(U)=0$ iff $U$ is a local product operator.

\subsection{Entangling power and a complementary quantity}

Notice that $E(U)$ and $E(US)$ are in some sense complementary quantities, as for a
product operator,
\beq
E(u_{A} \otimes u_{B})=0,\; \text{while}\; E((u_{A} \otimes u_{B})
S)=E(S)=1-\frac{1}{N^2}.
\eeq
The last relation follows from the Schmidt decomposition of $S$ which 
is 
\beq
S= \sum_{i,k=1}^{N} e_{ik} \otimes e_{ki}, \; \text{where}\; e_{ik}=|i\kt \br k|.
\eeq
Here $\{|i\kt,\, 1\le i \le N\}$ denotes
any orthogonal basis and hence represents a continuous family of possible Schmidt
decompositions, each with $\lambda_j=1$ for $1\le j\le N^2$.
The \textsc{swap} operator has the maximum operator entanglement entropy according to any measure
of entropy, including the linear one, $E(S)$ as above. Thus if $E(U)=E(S)$ the operator is $U$ maximally entangled. The complementary quantity $E(US)$ vanishes in the case of the \textsc{swap} gate, that is, $E(US)=E(S^2)=0$. 
In fact, linear combinations of these two complementary quantities give rise to two measures that are extensively discussed in this paper. 

One of the two measures we look at is the well-studied ``entangling power". The entangling power $e_p(U)$ \cite{Zanardi2000,Zanardi2001} of an operator
 $U\in \mathcal{H}^N_A\otimes\mathcal{H}^N_B$ 
is defined as the average entanglement created when $U$ acts on product state
$|\psi_A\kt|\psi_B\kt$ sampled according to the Haar measure on the individual spaces:
\beq
\label{eq:ep_symm}
e_p(U) = \left(\frac{N+1}{N-1} \right)\,
\overline{\mathcal{E}(U|\psi_A\kt|\psi_B\kt)}^{\psi_A,\psi_B}.
\eeq
Here the entanglement measure is the linear entropy $\mathcal{E}(|\psi\kt) = 1 -
\tr_A(\rho^2_A)$ and $\rho_A$ is the reduced density matrix $\tr_B(|\psi\kt\br\psi|)$. It
has been shown in \cite{Zanardi2001} that for any gate $U$ its  entangling power
can be expressed by the linear operator entanglement entropy,
\beq
e_p(U)=\frac{1}{E(S)}\left[E(U) + E(US) - E(S)\right]. 
\label{eqn:ep}
\eeq
The range of $e_p(U)$ is,
\beq\label{eq:eplimits}
0\leq e_p(U) \leq 1,
\eeq
which follows from the fact that the maximum value of $E(U)$ is $E(S)$. 
We have rescaled the definition of $e_p(U)$ from that originally defined in \cite{Zanardi2000},
 so that the maximum value is simply $1$ independent of $N$.

If $e_p(U)=0$ then $U$ is either a product of local operators or locally equivalent to the \textsc{swap}. The fact that \textsc{swap} does not create
any entanglement when acting on product states leads to $e_p(S)=0$, but that it is highly
nonlocal is reflected in its operator entanglement being maximum. This is one motivation
for introducing the complementary quantity
\beq\label{eqn:gt}
g_t(U) := \frac{1}{2E(S)}\left[E(U) - E(US) + E(S) \right]
\eeq
where $g_t$ is referred to as \textsl{gate typicality} in \cite{JMZL2017}. 
The range of $g_t(U)$ is 
\beq
0 \le g_t(U) \le 1,
\eeq
and $g_t(U)=1$ iff $U$ is the \textsc{swap} or is locally equivalent to the \textsc{swap}. Again, we have rescaled $g_t$ from the original definition in \cite{JMZL2017} by a factor of $2$ for
complete parity with $e_p$.

Thus while $e_p$ does not distinguish the local operators from the \textsc{swap}, $g_t$ does. It
turns out that rather than discussing the pair $\{E(U), E(US)\}$ in several settings
it seems more natural to work in the plane 
$\{e_p(U), g_t(U)\}$. 
The average of these measures when $U$ is sampled uniformly from the space of unitary
matrices with respect to the Haar measure 
constitutes the average over the circular unitary ensemble (CUE)
 and reads
\beq
\label{eq:CUEavgs}
\overline{E}=\overline{e_p}=\frac{N^2-1}{N^2+1}=\frac{E(S)}{2-E(S)}, \;\;
\overline{g_t}=\frac{1}{2}.
\eeq
As the scale is set so that $g_t\in [0,1]$,
and the Haar average reads $1/2$
we see that both classes of local gates and gates locally equivalent to 
 \textsc{swap} are equally non-typical.
The fact that $\overline{E}$ and $\overline{e_p}$ are close to the maximal possible value indicates that a typical Haar unitary gate 
has strong entangling properties \cite{Kus2013},
in analogy to the known fact that a generic bipartite pure state
is strongly entangled \cite{Lubkin1993,ZS01}.

Computation of the operator Schmidt decomposition and the operator entanglements follows from suitable permutations of the unitary matrix. If $\br i \alpha|U|j \beta \kt =\br i j |U^R|\alpha \beta\kt$ and 
$\br i \alpha|U|j \beta \kt =\br j \alpha |U^{T_A}|i \beta\kt$ denote the reshuffling (also referred to as realignment)
 and the partial
transpose operations respectively, we may define the following density matrices  \cite{Zyczkowski2004}:
\beq
\label{eq:rhoRandTdefn}
\rho_R(U)=\frac{1}{N^2} U^R U^{R \dagger}, \;\; \rho_T(U)=\frac{1}{N^2} U^{T_A} U^{T_A \dagger}.
\eeq  
Their linear entropies are given by
\beq
\label{Eq:EUUR}
E(U)=1-\frac{1}{N^4} \tr \left( U^R U^{R \,\dagger}\right)^2, 
\eeq
and 
\beq
\label{Eq:EUSUT}
E(US)=1-\frac{1}{N^4} \tr \left( U^{T_A} U^{T_A \,\dagger}\right)^2. 
\eeq
The operational interpretations of these quantities in terms of state entanglement of the equivalent $4$-party system is 
elaborated in Appendix \ref{app:4party}, including the generalization to the case of unequal subsystem dimensionality.

Note that if $U^R$ is also unitary then $E(U)=E(S)$ is the maximum possible. 
Unitary operators whose reshuffling is also unitary have recently been called {\sl dual-unitaries} due to their appearance in lattice models with space-time duality \cite{Akila2016,Bertini2019,Bertini2019b}. This class contains, for instance, the  discrete Fourier transform $F_{N^2}$,
for which all coefficients in the operator Schmidt decomposition are equal \cite{Nielsen2003,Kus2013}.
This dual-unitary property allows for special many-body systems built out of such unitaries to be solvable in some sense \cite{BKP2019,GBAWG19,PBCP20}, although they can be non-integrable. It is indeed interesting that such unitaries are also maximally entangled in the operator entanglement sense.
A way to generate ensembles of dual-unitaries has been presented in \cite{Suhail2020}.

If $U^{T_A}$ is also unitary, apart from $U^R$, then $E(US)=E(S)$.
 Such a matrix $U$,   called ``$2$-unitary" \cite{Goyeneche2015},
 saturates the maximum of entangling power, set to unity  by  our normalization.
As discussed in Appendix~\ref{app:4party} 
any two-unitary matrix of order $N^2$  corresponds to 
a two-uniform state   state of four quNits,
maximally entangled with respect to three possible symmetric partitions of the system \cite{HCLRL12}.
Any  dual-unitary, which satisfies weaker constraints, represents a 4-party state maximally entangled with respect to two possible partitions out of three.
Any unitary matrix of size  $N^3$, which remains unitary for any possible choice of three indices
out of six is called  three-unitary. It  maximizes the tri-partite entangling power \cite{Linowski_2020},
and  corresponds to a three-uniform state of six parties, maximally entangled
with respect to any splitting of the system into three plus three parts.
In general, a $k$-unitary matrix of size $N^k$ 
represents a $k$-uniform pure state of $2k$ subsystems  \cite{Goyeneche2015},
maximally entangled with respect to any symmetric partition of the system, 
 and therefore called  {\sl absolutely maximally entangled} (AME) state.

\section{Boundaries of two-qubit gates}
\label{sec:2qubits}

We focus on the two simple cases of  two-qubit and two-qutrit unitary gates.
 In particular, for $N=2$ and $N=3$ we study the structure of the set $K_N$
 of unitary matrices, $U(N^2)$, projected onto the plane $\{e_p(U), g_t(U)\}$.  
 Due to the normalization used the   phase-space is restricted to the square $[0,1]^2$.
 We will be interested in describing the boundary of the
 allowed area within the square
 and identifying particular gates corresponding to
 the distinguished points of the boundary.

The gate typicality $g_t$ and entangling power $e_p$  
 for two-qubit unitaries $U$, drawn at random from CUE$(4)$, are shown in Fig.~\ref{fig:ep_vs_gt}. It is clear that $0\leq e_p\leq 2/3$, reflecting the well-known fact that the maximum possible value of entangling power for a two-qubit gate is not $1$ (with our choice of factors), but is only $2/3$~\cite{Zanardi2000}. This is related to the nonexistence of absolutely maximally entangled states for a $4$-qubit system \cite{Higuchi2000}, as already mentioned above, and explained in Appendix~\ref{app:4party}.

Gate typicality is symmetric about its mean value $\overline{g_t(U)}^U = 1/2$ and 
this is reflected by the following equality,
    \beq \label{eq:gtrelation}
    g_t(U)+g_t(US) = 1.
    \eeq
Its maximal value  $g_t=1$ is attained only by the \textsc{swap} gate 
and its local equivalents, while the minimal value 
$g_t=0$ corresponds to local operators. Therefore, it might be appropriate to call the operators with $1/2\leq g_t\leq 1$,  \textsl{\textsc{swap}}-like. 
    
The boundaries of the set $K_2$ shown in Fig.~\ref{fig:ep_vs_gt} can be found using the limits of operator
entanglement $E(U)$ and $E(US)$. Writing these quantities in terms of the
entangling power $e_p$ and gate typicality $g_t$ of a two-qubit operator, $N=2$,
leads to
\beq
\begin{split}
E(U) &= \frac{3}{8}\left[e_p(U) + 2g_t(U)\right] \\
E(US) &= \frac{3}{8}\left[e_p(U) - 2g_t(U) + 2\right].
\end{split}
\eeq
  \begin{figure}
\includegraphics[width=0.5\textwidth]{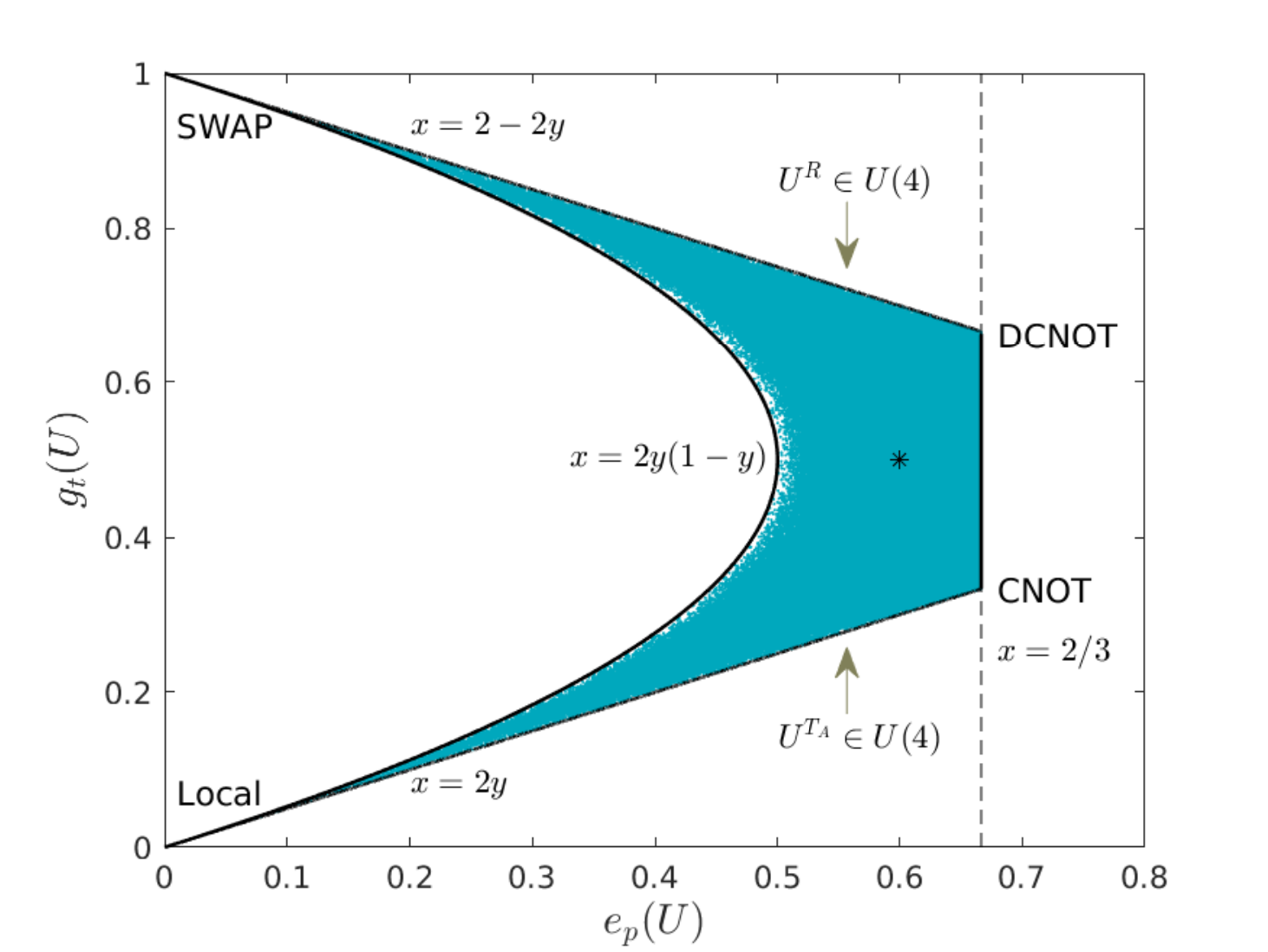}    \caption{
Set $K_2$ of all two-qubit gates projected into the plane  
entangling power, $x=e_p(U)$, 
versus gate-typicality, $y=g_t(U)$.
Each point corresponds to a random unitary matrix from $ U(4)$, 
while (*) represents the average over CUE.
The boundaries  $\partial K_2$ are identified in the text and in Fig.  \ref{fig:ep_gt_tedges}.
}
    \label{fig:ep_vs_gt}
    \end{figure}

The upper bounds on $E(U)$ and $E(US)$ (equal to $3/4$) lead to the relations,
 \beq
 \label{eq:linesepgt}
e_p+2g_t\leq 2 \quad \text{and} \quad g_t\geq e_p/2,
    \eeq
which are the top and bottom lines in Fig.~\ref{fig:ep_vs_gt}. The maximum value of $e_p=2/3$ is reached by the \textsc{cnot} gate and is an ``optimal" gate in the
terminology of \cite{Zanardi2000}. The region is further restricted however and we will show below that the left boundary is given by the parabola $e_p=2g_t(1-g_t)$.
We further show  in Sec.~\ref{sec:qudits} that this boundary in fact consists of gates 
of the form $S^{\alpha}$ with  $0\le \alpha \le 1$,
that are rational powers of the \textsc{swap} operator $S$.

\subsection*{The Weyl chamber and various gates}
     
%

While the lines in Eq.~(\ref{eq:linesepgt}) are bounds, we identify the gates that make
these actual boundaries of the allowed  set $K_2$ in the space $e_{p}$ vs $g_{t}$. 
It will be useful to work with the well known canonical form of a two-qubit unitary operator. Any
two-qubit operator $U\in SU(4)$\footnote{The statement extends to any $U \in U(4)$, since any bipartite unitary $U \in U(4)$ can be expressed as the product of a $U \in SU(4)$ and a global phase shift $e^{i\alpha}$.}, upto left and right multiplication by local unitaries,
can be expressed in terms of Euler angles $\{c_{1}, c_{2}, c_{3}\} \in [0,\pi]$
as~\cite{KBG01,KC01,Zhang2003, Rezakhani2004},
\begin{equation}
\label{can2x2}
\begin{split}
U =\exp \left[-i \left(\frac{c_{1}}{2}\sigma_{1}\otimes\sigma_{1} + \frac{c_{2}}{2}
\sigma_{2}\otimes\sigma_{2} + \frac{c_{3}}{2}\sigma_{3}\otimes\sigma_{3}\right)\right],
\end{split}
\end{equation}
where $\{\sigma_{1}, \sigma_{2}, \sigma_{3}\}$ are the Pauli matrices. In the standard
computational basis (the eigenbasis of $\sigma_{3}$), any bipartite unitary operator can
thus be written as,
\begin{equation}
    U =     \left( \begin{array}{cccc}
    e^{\frac{-ic_3}{2}}c^-       & 0 & 0  & -ie^{\frac{-ic_3}{2}}s^- \\
    0       & e^{\frac{ic_3}{2}}c^+ & -ie^{\frac{ic_3}{2}}s^+ & 0 \\
    0       & -ie^{\frac{ic_3}{2}}s^+ & e^{\frac{ic_3}{2}}c^+ & 0 \\
    -ie^{\frac{-ic_3}{2}}s^-       & 0 & 0  & e^{\frac{-ic_3}{2}}c^-  
\end{array} \right) \nonumber
\end{equation}
where,
\beq
c^\pm = \cos[(c_1\pm c_2)/2];\qquad s^\pm = \sin[(c_1\pm c_2)/2].
\eeq
On imposing the constraint of local unitary equivalence, that is, if any two unitaries
$U$ and $U' = (u_{A_1} \otimes u_{B_1}) U (u_{A_2} \otimes u_{B_2})$ related by local
unitaries are represented by the same set of Euler angles, the range of values gets
restricted to $|c_{3}| < c_{2} < c_{1} < \pi/2$. This
region in the $\{c_{1},c_{2},c_{3}\}$ space containing the nonlocal two-qubit gates forms a tetrahedron known as the Weyl chamber~\cite{Zhang2003}. 

In terms of the $\{c_{1}, c_{2}, c_{3}\}$ parametrization, it is known~\cite{Zhang2003,
Balakrishnan2009} that one can define two quantities which are invariant under local unitary
operations, namely,
\begin{align}
\begin{split}
G_1 &= \cos^2c_1\cos^2c_2\cos^2c_3 - \sin^2c_1\sin^2c_2\sin^2c_3 \\
& +  \frac{i}{4}\sin2 c_{1}\sin 2c_{2}\sin 2c_{3},\\
G_2 &= \cos2c_1 + \cos2c_2 + \cos2c_3.
\end{split}
\end{align}
The operator entanglements $E(U)$ and $E(US)$ can be written in terms of local invariants
$G_1$ and $G_2$, as follows~\cite{Balakrishnan2011}:
\begin{align}
\begin{split}
E(U) &= 1 - \frac{1}{8}\left[3 + 2|G_1(U)| + G_2(U)\right],\\
E(US) &= 1 - \frac{1}{8}\left[3 + 2|G_1(U)| - G_2(U)\right].
\end{split}
\end{align}
Consequently, the entangling power and gate-typicality of any two-qubit gate $U$ can be
explicitly evaluated in terms of the angles $\{c_{1}, c_{2}, c_{3}\}$ and takes on an
elegant and simple form as,
\begin{eqnarray}
e_p(U) &=& \frac{2}{3}\left[ \sin^{2}c_{1}\cos^{2}c_{2} + \sin^{2}c_{2}\cos^{2}c_{3} +
\sin^{2}c_{3}\cos^{2}c_{1} \right], \nonumber \\
g_{t}(U)&=& \frac{1}{3}\left[ \sin^{2}c_{1} + \sin^{2}c_{2} + \sin^{2}c_{3} \right].
\label{eq:ep_gt_formula}
\end{eqnarray}
This leads to the following restriction on the allowed region in the $e_{p}-g_{t}$ plane
for two-qubit gates.
\begin{theorem}[Boundary of two-qubit gates]
The entangling power $e_{p}(U)$ and gate-typicality $g_{t}(U)$ for any two-qubit unitary
$U$ satisfy
\begin{equation}
\label{eq:epgtparabola}
e_{p}(U) \geq 2g_{t}(U)\left( 1 - g_{t}(U)\right).
\end{equation}
\end{theorem}
{\it Proof:} Using Eq.~\eqref{eq:ep_gt_formula}, we see that $2 g_{t}(U) (1-g_{t}(U))$ is of
the form,
\begin{eqnarray}
 \label{eq:parabola1}
 \! \! \! \! \!\! \!\!
 2g_{t}(U)(1-g_{t}(U))=\frac{2}{9}( x+ y + z) (3 - (x + y + z)),
\end{eqnarray}
where $x \equiv \sin^{2}c_{1}$, $y \equiv \sin^{2}c_{2}$, $z \equiv \sin^{2}c_{3}$
satisfy $0 \leq x,y,z \leq 1$. Then, it is easy to see that,
\begin{eqnarray}
&& ( x+ y + z) (3 - (x + y + z)) \nonumber \\
&=& 3\left[x(1-y) +y(1-z) + z(1-x)\right] \nonumber \\
&+& \left(xy + yz + zx \right) - \left(x^{2} + y^{2} + z^{2}\right) \nonumber \\
&\leq& 3\left[x(1-y) +y(1-z) + z(1-x)\right], \nonumber
\end{eqnarray}
since $xy + yz + zx \leq x^{2} + y^{2} + z^{2}$, by Schwarz inequality. Using this in
Eq.~\eqref{eq:parabola1} above, we get,
\begin{eqnarray}
&& 2g_{t}(U)(1-g_{t}(U)) \nonumber \\
&\leq&  \frac{2}{3} \left[x(1-y) +y(1-z) + z(1-x)\right] \nonumber \\
&=& \frac{2}{3}\left[ \sin^{2}c_{1}\cos^{2}c_{2} + \sin^{2}c_{2}\cos^{2}c_{3} +
\sin^{2}c_{3}\cos^{2}c_{1} \right] \nonumber \\
&=& e_{p}(U),
\end{eqnarray}
as desired. \qedsymbol

\begin{table}[h!]
  \centering
  \begin{threeparttable}
    \begin{tabular}{ l  p{1.5cm}  p{1.5cm} p{1.5cm} p{1.5cm}}
    \hline\hline
    Gate $U$				& $E(U)$ 		& $E(US)$			& $e_p(U)$ 			& $g_t(U)$
    
    \\
\hline
    Local-gate 			& 0 				& $\frac{3}{4}$ 	&0 					&0\\[2ex]
$\sqrt{\textsc{cnot}}$ & $\frac{1}{4}$ & $\frac{3}{4}$ &$\frac{1}{3}$
&$\frac{1}{6}$\\[2ex]
\textsc{cnot, B-gate} & $\frac{1}{2}$ & $\frac{3}{4}$ &$\frac{2}{3}$
&$\frac{1}{3}$\\[3ex]
\textsc{dcnot} & $\frac{3}{4}$ & $\frac{1}{2}$ &$\frac{2}{3}$ &$\frac{2}{3}$\\[2ex]
Fourier  $F_4$			& $\frac{3}{4}$ 		& $\frac{1}{4}$ 		&$\frac{1}{3}$ 		&$\frac{5}{6}$\\[3ex]$\sqrt{\textsc{swap}}$ & $\frac{9}{16}$ & $\frac{9}{16}$ &$\frac{1}{2}$
&$\frac{1}{2}$\\[2ex]
       \textsc{swap} 		& $\frac{3}{4}$ 		& $0$ 			&$0$ 				&$1$
       
        \\
\hline
\\
Haar Average &$\frac{3}{5}$ &$\frac{3}{5}$			&$\frac{3}{5}$ 			&$\frac{1}{2}$

\\
\hline\hline
    \end{tabular}
  \end{threeparttable}
    \caption{Nonlocal properties of selected two-qubit gates, $N=2$. 
    Their location in the set $K_2$ is shown in Fig. \ref{fig:ep_gt_tedges}.}
    \label{table:5.2}
\end{table}

\begin{figure}[!tbp]
 \centering
\includegraphics[width=0.5\textwidth]{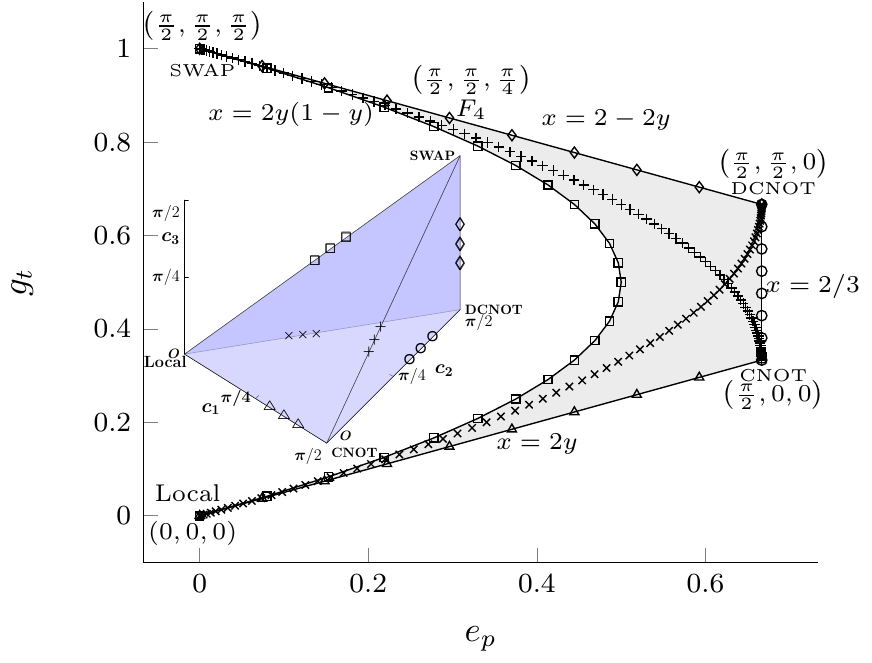}
    \caption{Boundaries of the set $K_2$ representing two-qubit gates 
    in the plane $(e_p, g_t)$
are indicated by solid lines. Note distinguished gates  identified in the plot. 
Inset shows edges of the
tetrahedron  in the parameter space $(c_1,c_2,c_3)$  forming a half of the Weyl chamber \cite{Mandarino2018},
   which correspond to  $\partial K_2$.
   }
    \label{fig:ep_gt_tedges}
  \end{figure}

The inequality in Eq.~(\ref{eq:epgtparabola}) is tight, as the $S^\alpha$ family of gates
with $0\leq\alpha\leq1$ lie on the parabola $e_{p}(U) = 2g_{t}(U)\left( 1 - g_{t}(U)\right)$. This is shown by an explicit calculation in Eq.~(\ref{eq:swapparab}), Sec.~\ref{sec:qudits}.

The \textsc{cnot} gate $C$ has the maximum entangling power of $2/3$, as expected. Furthermore, we show below that all
members of the $C S^\alpha$ family, with $0 \leq \alpha \leq 1$ have maximum entangling
power of $2/3$ and form the rightmost vertical boundary in
\Cref{fig:ep_vs_gt,fig:ep_gt_tedges}. The gate $CS$ is the so-called double-\textsc{cnot} ({\sc dcnot})
gate \cite{Collins2001}. Note that, 
\beq
U_t=\exp( i t S)= \mathbb{1}\, \cos t + i \sin t \, S, \label{eq:ut}
\eeq
as $S^2=\mathbb{1}$, where, $\mathbb{1}$ denotes the identity operator. This is a route to defining fractional powers of $S$, as
$\exp(i \pi S/2) = iS$ and therefore
$(iS)^{t 2 /\pi}$ is same as $\exp(i t S)$ and the overall
phase of $i^{t 2 /\pi}$ makes no difference to any of the subsequent 
calculations. Therefore $U_t$ is essentially $S^{2t/\pi}$.  
The reshuffled matrix of $CS^\alpha$ is upto a constant phase given by,
\beq
\label{eq:CSalphaR}
(CS^\alpha)^R = \cos(\pi \alpha/2) C^R+ i \sin (\pi \alpha/2) (CS)^R. 
\eeq
The rearrangement
of the \textsc{cnot} gate is non-unitary being $|00\kt \br 00|+|00\kt \br 11|+|11\kt \br
01| +|11\kt \br 10|$, while $(CS)^R$ is again a permutation given by $|00\kt \br
00|+|10\kt \br 11|+|01\kt \br 10|+|11\kt \br 01|$. A calculation then yields that
\beq
E(CS^{\alpha})=\frac{1}{8}(5-\cos(\pi \alpha))
\eeq
Hence 
\beq
e_p(CS^{\alpha})=\frac{1}{E(S)}\left[ E(CS^{\alpha})+E(CS^{\alpha+1})-E(S)\right]=\frac{2}{3},
\eeq
and $g_t(CS^{\alpha})=1/2 -\cos(\pi \alpha)/6$ interpolating between $1/3$ and $2/3$.

Several other standard two qubit gates are identified and their operator entanglement and
entangling powers are given in Table~\ref{table:5.2}.
We also identify gates in the Weyl chamber with different regions of the set $K_2$
contained in the plane  $(e_p, g_t)$.
In  \Cref{fig:ep_gt_tedges}, six edges of the tetrahedron
forming a half of the chamber \cite{Mandarino2018}
are shown. Four of these edges form four of the boundaries $\partial K_2$,
the other two connect two of the extreme points symmetrically.


\section{Beyond qubits and the entangling power of some quNit gates}  
\label{sec:qudits}
Moving beyond qubits, we now study the entanglement landscape of bipartite unitary gates
acting in a composite $N \times N$ quantum system.
 In this context, we investigate the Fourier gate and the fractional powers of the \textsc{swap} that form an important family of gates. We observe that for any $N \geq 2$, the fractional powers of \textsc{swap} lie on a parabola. The rightmost point is maximally entangling, at $e_p=1$ and $g_t=1/2$ and it is known that in all dimensions except $N=6$ (and $N=2$, which we have already dealt with) permutations exist which have
these values. In the case $N=3$ explicit examples of permutations which have $e_p=1$ have been constructed \cite{Clarisse2005,Goyeneche2015}.


The discrete Fourier transform, DFT, on the space
$\mathcal{H}_N \otimes \mathcal{H}_N$
is given by the unitary gate $F_{N^2}$ of order $N^2$,
 with entries  $F_{mn}=\frac{1}{N}\exp(2 \pi i mn/N^2)$. 
 This may be expressed in bipartite
notation, as, 
\beq
\br k \alpha |F_{N^2}|j \beta \kt = \frac{1}{N}e^{\frac{2 \pi i}{N^2}[(k+\alpha N)(j +\beta
N)]},
\eeq
where $0 \leq k,j, \alpha, \beta \leq N-1$. It is then straightforward to verify that the
reshuffled matrix $F_{N^2}^R$ is also unitary \cite{Kus2013}, and hence the operator
entanglement is maximum possible: $E(F_{N^2})=1-1/N^2$. In this sense the Fourier gate in
arbitrary dimensions is a dual-unitary, and a recent paper \cite{GBAWG19} constructs dual kicked chains using the DFT, to study solvable Floquet many-body systems.

However, the partial transpose of the DFT is not unitary and
hence the Fourier does not have maximal entangling power. Equivalently $E(F_{N^2}S)$ is not the maximal possible, instead a calculation
yields
\beq
\begin{split}
E(F_{N^2}S)=&1-\frac{1}{N^4}\left[N^3+2\sum_{k=1}^{N-1}k\frac{\sin^2\left(k\,\pi/N\right)}{\sin^2\left(\pi/N-k\,\pi/N^2\right)}\right]\\
& \approx 1-\frac{2}{\pi^2} \int_0^{1} \frac{x \sin^2(\pi x)}{(1-x)^2}\, dx \approx
0.344,
\end{split}
\eeq
where the approximation is valid for large $N$. Thus the operator entanglement of $F_{N^2}S$
and the entangling power of the Fourier gate $F_{N^2}$
tends to $\approx 0.344$, about one-third of the maximum possible.

As indicated in Eq.~\eqref{eq:ut} above,  the fractional powers
of the \textsc{swap} $S$ up to phase factors are given by $U_t=\exp(i t S)$. Since the reshuffled operator $S^R=S$, we
get
\begin{eqnarray}
U_t^R &=& \mathbb{1}^R \cos t + i \sin t\,  S \nonumber \\
&=& N |\Phi^+ \kt \br \Phi^+| \cos t+ i \sin t \, S,
\end{eqnarray}
where we have use the fact that the reshuffling of the identity is given by, $\mathbb{1}^R= N |\Phi^+ \kt \br \Phi^+|$, with $|\Phi^+\kt =
\frac{1}{\sqrt{N}}\sum_{i=1}^N|ii \kt$ being a maximally entangled state. Further, as $U_t\, S= S\,  \cos t +i \mathbb{1} \sin t$, 
the following simple formulae follow for the fractional powers of the \textsc{swap} gate:
\beq
\label{eq:swapparab}
\begin{split}
&E(e^{i t S})=E(S) (1-\cos^4 t),\; E(e^{itS}S)=E(S)(1-\sin^4 t),\\
&e_p(e^{itS})=\frac{1}{2} \sin^2(2t), \; g_t(e^{itS})=\sin^2 t.
\end{split}
\eeq
Thus, if $U_t$ is a
fractional power of $S$ then  $e_{p}(U_t)=2 g_t(U_t) \left(1-g_t(U_t) \right)$, in any dimension. We have already shown that this parabola is indeed the left-boundary of the  set $K_2$ in the $(e_{p},g_{t})$ plane in the case of two-qubit gates.

%
%

To investigate the neighborhood of the parabola, we start with an operator of the form $S^{\alpha}$ and perturb it, while retaining the unitarity. There are many possible ways of doing such a perturbation, all of which yield equivalent results. For example one may deform $S^{\alpha} \rightarrow S^{\alpha}
\exp(i \epsilon H)$ where $H$ is a random Hermitian matrix with unit variance and zero mean elements. Another approach is to 
use random matrices
$ V_{\epsilon}= U_{CUE} U_d(\epsilon) U_{CUE}^{\dagger}$  from the ensemble investigated in \cite{PZK98}
and  defined by a Haar random unitary matrix  $U_{CUE}$ and a diagonal matrix $U_d(\epsilon)$
with  phases $\exp(i \epsilon \xi)$, where $\xi$ is uniform random number in $[-\pi, \pi )$.
Powers of  \textsc{swap} perturbed as $S^{\alpha} \rightarrow S^{\alpha} V_{\epsilon}$
result in values of $\{e_p, g_t\}$ lying to the right of the parabola. Combined with the stationarity derived in Appendix~(\ref{app:lemmaswap}), one may be tempted to conjecture that the parabola itself is a boundary. However, we have found an exception in a permutation in the qutrit case and can only conclude that typical perturbations of $S^{\alpha}$ result in a movement to the right of the parabola in the ($e_p$, $g_t$) plane.

\begin{figure}[!tbp]
\includegraphics[width=.53\textwidth]{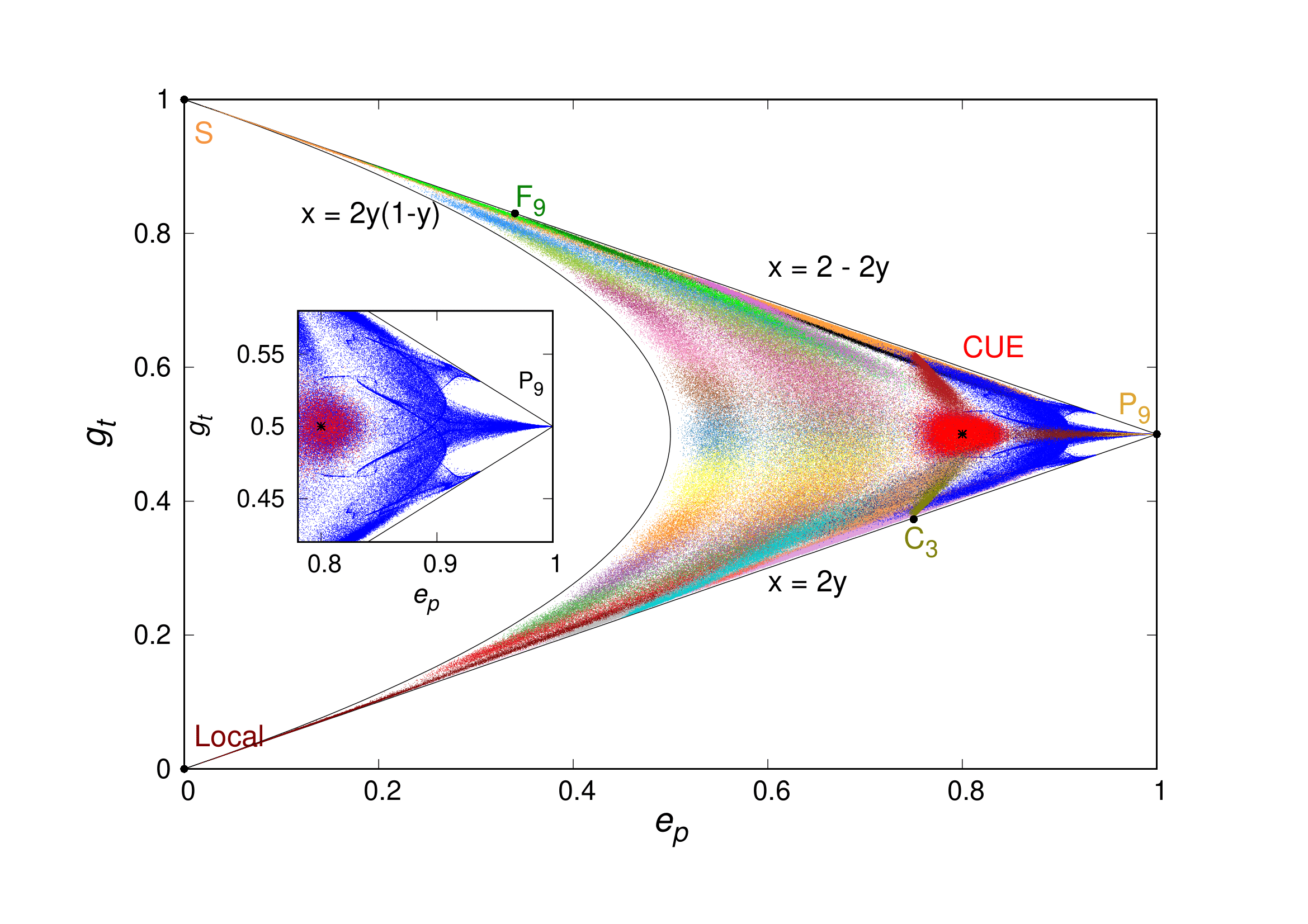} 
\caption{(color online:) Unitary matrices $U\in U(9)$, representing two-qutrit gates, projected into the set $K_3$
 in the  plane $(e_p, g_t)$. 
Each color represents the neighbourhood of a particular gate, labeled with the same color.
Upper side of the triangle, including Fourier matrix $F_9$, contains 
`dual unitaries' for which $U^R$ is unitary, while the lower side
includes these for which $U^{T_A}$ is unitary.
 Inset provides a magnified view of the region around the rightmost point representing the $2$-unitary gate $P_{9}$, 
  with the CUE `cloud',  centered at $(1/2,4/5)$  (*), and shown in red
and perturbations of different boundary gates 
 in blue. 
}
  \label{fig:epgt_qutrit}
  \end{figure}

A similar study as in the case of $K_{2}$ was performed for unitary matrices belonging to the lower and the upper parts of the boundary of the set $K_3$. It is useful to distinguish
certain unitary matrices, which correspond to points at  $\partial K_3$.
The controlled addition gate $C_N$ acting on a two-quNit system can be considered as a generalizations of the standard CNOT gate.
In the case of $N=3$ such a gate reads,
\beq
\   C_3 |i\kt\otimes |j\kt = |i\kt\otimes |i\oplus j\kt, \,\, i,j\in \mathbb{Z}_3,
\eeq
where $\oplus$ denotes addition modulo 3.
This gate attains the  maximal value of  $E(C_3\,S)$ and lies in $K_3$ on its lower boundary, $y=2x$. 
It is seen that the perturbations have the tendency to quickly approach the CUE ``cloud" in the manner of a jet.

In Fig.~\ref{fig:epgt_qutrit}, the neighbourhood gates of several unitary
quantum gates are generated for $N=3$ and the corresponding phase space plot is shown. The rightmost point 
of the set $K_3$ 
in the $(e_{p},g_{t})$ plane, denoted as $P_9$ in Fig.~\ref{fig:epgt_qutrit}, corresponds to one of the permutations with $e_p=1$ defined in \cite{Clarisse2005, Goyeneche2015}. The Fourier matrix $F_9$,  attains the maximum value of $E(U)$, as $F_9^R$ is unitary, 
and lies on the upper boundary of $K_3$ formed by the line $x=2(1-y)$. 

The upper boundary line contains  maximally entangled 
unitary matrices,  for which $U^R$ is also unitary.
However,  the partially transposed matrix $U^{T_a}$ is not unitary, with the exception of the matrices at the right corner of the triangle. Thus  gates belonging to the upper boundary of $K_N$ are not $2$-unitary \cite{Goyeneche2015},
but satisfy the weaker condition of being dual-unitaries \cite{Bertini2019b}. Unitary gates for which $U^{T_A}$ is unitary,  studied  in \cite{DNP16,BN17} in context of quantum operations preserving some given matrix algebra, belong to the lower boundary line of $K_3$. Both lines cross at the right corner of the triangle, representing permutation $P_9$ and other $2$-unitary matrices, which maximize the entangling power.


It is interesting to observe that the set $K_3$ seems not to fill entire edge of the triangle
close to the corner with $e_p=1$, as no dual unitaries
in the vicinity of $P_9$ were found. 
This fact is borne out by numerical simulations that employ an algorithm to create an ensemble of dual ones \cite{Suhail2020}. The significance of the gap observed is to be fully explored, but the numerics suggest that the set of dual unitary matrices of size $N^2=9$ is not connected, in contrast to the two-qubit case, $N^2=4$. Since the dual unitary operators are related to four-party entanglement  -- see Appendix~\ref{app:4party} --
this implies some additional constraints on the entanglement in four-qutrit systems across different partitions and on possible spectra of two-partite density matrices obtained by partial trace of a pure state of size $N^4$.

Analysis of the non-local properties of any two-qubit gate
becomes easier as the canonical form (\ref{can2x2})
is valid for any unitary matrix from $U(4)$.
This  form, related to a isomorphism in group theory between $SO(4)$
and $SU(2) \times SU(2)$ can not be generalized for two-qutrit gates.
Therefore, our understanding of the set of bipartite gates acting on $N \times N$
systems in still not complete. The structure of the set $K_N$
obtained by a projection of $U(N^2)$ into the plane $(e_p, g_t)$ 
is not entirely characterized even in the case $N=3$.
Leaving these open problems for further studies
we shall now move to a related problem,
if a given bipartite unitary gate $U$ acts sequentially on a quantum system.

\section{Time evolution and multiple uses of the nonlocal operators}
\label{sec:thermal}

If $U$ is a bipartite quantum propagator, it is natural to consider a combination
${\mathcal U}= (u_{A} \otimes u_{B}) U$ where the unitaries $u_{A},u_{B}$ are interpreted as ``local dynamics" or single particle dynamics. 
We have motivated (see discussion around Eq.~(\ref{eq:Upower(n)})) the study of its powers $\mathcal{U}^n$ as well as products $\mathcal{U}^{(n)}$ with different local operators in each term of the product.

The circuit in Fig~\ref{fig:circuit} describes the time-evolution scenario considered here, for the case of qubit systems. Specifically, the circuit depicts the propagator $\mathcal{U}^{(n)}$ for $n=3$. The fixed nonlocal unitary $U \in U(4)$ is implemented via a combination of $\textsc{cnot}$ gates and local rotations $R_{z}(t)$ and $R_{y}(t)$, following the prescription in~\cite{vatan}. The interlacing local qubit gates are denoted as $A_{i}$ and $A_{i}'$, with $i=1,2$. We have omitted the initial set of local unitaries since they do not affect the entangling power. Note that the interlacing locals are different at each step, and hence labelled differently.

Observe that for a single time step the nonlocal content of $U$ is the same as that of ${\mathcal U}$,  hence $e_p(U)=e_p({\mathcal U})$. 
Thus if the gate ${\mathcal U}$  is applied onto an unentangled initial state
the local dynamics does not play any role in creation of quantum entanglement.
However, the nonlocal content of multiple applications, either as ${\mathcal U}^n$ or $\mathcal{U}^{(n)}$, which represents discrete time evolution, is a different matter as the Schmidt coefficients of an operator
in general change on taking powers.   In this case the local dynamics can play a crucial
role \cite{JMZL2017, Mandarino2018}. For instance, in terms of entangling power
\beq
e_p({\mathcal U}^2)=e_p \left[ U (u_{A} \otimes u_{B}) U \right] \ne e_p(U^2).
\eeq
One of the aims of this paper is to analyze this difference
and study the regime of large $n$. While we have presented
related results earlier \cite{JMZL2017}, this work
contains an important generalization and a more elegant derivation that uses group
theory. 
Note that we are interested in generic statements about average entanglement growth
in time, a subject that already has a considerable literature and is still a topic of research.

\onecolumngrid

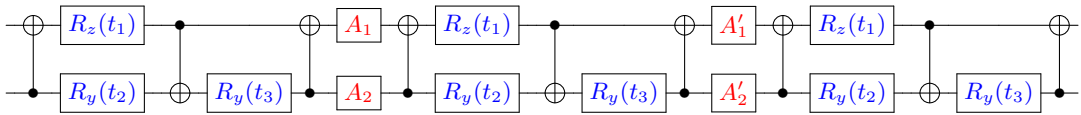
\begin{figure} 
\hspace{0.6cm}
\Qcircuit @C=0.35em @R=1.2em { 
 & \qw &\targ &\qw & \gate{\textcolor{blue}{R_{z}(t_{1})}} & \qw & \qw & \ctrl{1} & \qw & \qw & \targ  &\qw & \gate{\textcolor{red}{A_{1}}} & \qw &\targ &\qw & \gate{\textcolor{blue}{R_{z}(t_{1})}} & \qw & \qw & \ctrl{1} & \qw & \qw & \targ  & \qw & \gate{\textcolor{red}{A'_{1}}} & \qw &\targ &\qw & \gate{\textcolor{blue}{R_{z}(t_{1})}} & \qw & \qw & \ctrl{1} & \qw & \qw & \targ & \qw \\
 & \qw &\ctrl{-1} & \qw & \gate{\textcolor{blue}{R_{y}(t_{2})}} & \qw & \qw & \targ &\qw & \gate{\textcolor{blue}{R_{y}(t_{3})}} & \ctrl{-1} & \qw & \gate{\textcolor{red}{A_{2}}} & \qw &\ctrl{-1} & \qw & \gate{\textcolor{blue}{R_{y}(t_{2})}} & \qw & \qw & \targ &\qw & \gate{\textcolor{blue}{R_{y}(t_{3})}} & \ctrl{-1}   & \qw  & \gate{\textcolor{red}{A'_{2}}} & \qw &\ctrl{-1} & \qw & \gate{\textcolor{blue}{R_{y}(t_{2})}} & \qw & \qw & \targ &\qw & \gate{\textcolor{blue}{R_{y}(t_{3})}} & \ctrl{-1}   & \qw  \\
}  
\caption{Circuit representing the propagator $\mathcal{U}^{(3)}$ for a pair of qubits. The fixed nonlocal two-qubit gate $U$ is implemented via the sequence of $\textsc{cnot}$s and rotation gates $R_{z}(t)$,$R_{y}(t)$. The gates $A_{i}$ and $A_{i}'$ represent the interlacing local unitaries, which vary at each time step.} \label{fig:circuit}
\end{figure}
\twocolumngrid

\subsection{Thermalization of entangling power}

The generalization allows for the subsystems $\mathcal{H}^A$ and $\mathcal{H}^B$ to have different dimensions $N$ and $M$, say $N \leq
M$. The operator entanglement still follows from the Schmidt decomposition of $U$ as in
Eq.~(\ref{eq:SchmidtU}) and is determined
by the singular values of the reshuffled  matrix $U^R$ of size $N^2\times M^2$.
This gives the vector of local invariants $\lambda_j$,
equal to eigenvalues of a positive matrix $U^R (U^R)^{\dagger}$.
The other set of invariants which in the
symmetric case came from the Schmidt decomposition of $US$ in Eq.~(\ref{eq:SchmidtUS}) now come from the singular values of
the square matrix $U^{T_A}$ of size $NM\times NM$. The generalization of the expressions in (\ref{eqn:ep}) and \eqref{eqn:gt} for entangling power \cite{Wang2002} and gate-typicality, respectively, based on the reshuffled and partially transposed matrix $U$ is given by:
\beq
\label{eq:gen_ep}
\begin{split}
e_p(U)=&\frac{1}{M^2(N^2-1)} \left[NM(NM+1)- \right. \\ & \left. \tr(U^R(U^R)^\dagger)^2
-\tr(U^{T_A}(U^{T_A})^\dagger)^2 \right],\\
g_t(U) = & \frac{1}{2NM(N+1)(M-1)}\left[N^2M^2 - NM -\right.\\ &
\left. \tr(U^R(U^R)^\dagger)^2 + \tr(U^{T_A}(U^{T_A})^\dagger)^2 \right]
\end{split}
\eeq
Note that we use a normalization factor that implies that the maximal
 entangling power is equal to unity, which is attained when 
$\tr(U^R(U^R)^\dagger)^2=M^2$ and $\tr(U^{T_A}(U^{T_A})^\dagger)^2=NM$. Hence our expression differs from the expression in \cite{Wang2002} 
by a factor $\tilde{e}_p^{\text{max}}= \frac{M(N-1)}{N(M+1)}$,
which is the unscaled maximum entangling power for a $N \times M$ bipartite system.

The generalization in Eq.~\eqref{eq:gen_ep} allows us to consider a situation where the bipartite interaction is non-zero but arbitrarily small and the second subsystem is considerably large, such as a thermal bath. In particular, we show in Theorem~\ref{thm:ep_avg} below, that
\beq
\label{eq:epUVavg}
\begin{split}
& \left \br e_p\left[ U (u_A\otimes u_B) V \right]\right \kt_{u_A, u_B} \\ &=
e_p(U)+e_p(V) -\, e_p(U)e_p(V)/\overline{e_p},\\
& \left \br g_t\left[ U (u_A\otimes u_B) V \right]\right \kt_{u_A, u_B} \\ &=
g_t(U)+g_t(V) -\, g_t(U)g_t(V)/\overline{g_t},
\end{split}
\eeq
Here $U$ and $V$ are any two unitary operators and the angular brackets indicate averaging over the local unitary operations with
$u_{A,B}$ sampled uniformly (Haar measure). The quantities $\overline{e_p}$ and $\overline{g_t}$ are the Haar averages over the $MN$ dimensional space that generalize the expressions in Eq.~(\ref{eq:CUEavgs}), for subsystems of equal dimensions, \blue{to}
\beq
\label{eq:generalepgtavg}
\overline{e_p} = \frac{N(M^2-1)}{M(NM+1)}, \, \ \ 
\overline{g_t} = \frac{ (N-1)(M+1)}{2(NM - 1)}.
\eeq

The above discussion can be directly related to \emph{operator scrambling}, which  measures
the spread of an initially localized operator \cite{Nahum2017,Chalker2018,Keyserlingk2018, Moudgalya2019}.
 In the simplest bipartite setting, operator scrambling can be characterized
 by analyzing to what extent initially local operators become non-local.
  In analogy to the entangling power, wherein the action 
of operators on initially unentangled states is measured \cite{Zanardi2000,Zanardi2001}, 
we may consider time evolution of initial product operators in the Heisenberg picture of quantum mechanics. 
Such an evolution is obtained as a special case of Eq.~\eqref{eq:epUVavg}, if $U=V^\dagger$. 
This may be interpreted as the average entangling power on conjugation 
of product operators with the bipartite operator $V$. 
As $e_p(V^{
\dagger})=e_p(V)$, our results imply that 
\beq
\label{eq:scrambpower}
\left \br e_p\left[ V^{\dagger} (u_A\otimes u_B) V \right]\right \kt_{u_A, u_B} =
e_p(V)\left[ 2-\dfrac{e_p(V)}{\overline{e_p}}\right].
\eeq
This provides a way to quantify the \emph{scrambling power} of bipartite unitary operators.
It would be interesting to generalize  such  a scrambling power for a multipartite setup, 
in analogy to the entangling power applied recently for several subsystems \cite{Linowski_2020}.

 
It might be surprising that simple relations \eqref{eq:epUVavg} exists for $e_p$ and $g_t$,
and this is due to the fact  that they concern the average values. 
Similar relations hold for averaged operator entanglements $\left \br E\left[ U (u_A\otimes u_B) V \right]\right \kt_{u_A, u_B}$ and 
$\left \br E\left[ U (u_A\otimes u_B) V \, S \right]\right \kt_{u_A, u_B}$,
 but they mix among themselves in a less transparent way.
Although the statements above concern averages over local unitaries,
 they provide some immediate insights. For instance, choosing $U$ and $V$ such that
\beq
e_p(U)+e_p(V) -\, e_p(U)e_p(V)/\overline{e_p} \; > \; e_p(UV),
\eeq
we infer that there exist local unitaries which enhance the entangling power beyond a serial application
 of $V$ and $U$.
Relations in Eq.~(\ref{eq:epUVavg}) can be used to iterate, by inserting independent
local operators between nonlocal operators. For example, one becomes
\beq
\begin{split}
&\left \br e_p\left[ U (u_A\otimes u_B) V (u'_A\otimes u'_B) W \right]\right \kt_{u_A,
u_B,u'_{A}, u'_B} \\
& = e_p(U)+e_p(V)+e_p(W)-[e_p(U)e_p(V)\\ 
&+ e_p(V)e_p(W)+e_p(W)e_p(U)]/\overline{e_p} \\
&+ e_p(U)e_p(V)e_p(W)/(\overline{e_p})^2.
\end{split}
\eeq
The above equation indicates a certain ``decoupling" that is induced by local dynamics. It is necessary
that the local operators at each product be independent,
else the correlations prevent such an expression. However,  
previous work suggests  \cite{JMZL2017} that they provide
a good approximation, also in the case if the matrices $u_A$, $u_A'$ and $u_B$, $u'_B$ are pairwise identical,
 provided they are Haar - typical random unitaries. 


We now formally state the result concerning the thermalization of the entangling power and gate typicality averaged over 
random local dynamics in the generalized setting of unequal dimensions of the subsystems. The final formulae remain the same as those displayed in \cite{JMZL2017}, indicating a certain
universality in them. However we present an alternate proof here, based on irreducible representations of the unitary group. Due to the technical nature of the proof, we present the details separately in Appendix \ref{app:details}.

\begin{theorem}\label{thm:ep_avg}
Let $U$ and $V$ be bipartite unitary operators on $\mathcal{H}^A_N\otimes\mathcal{H}^B_M$ and
$u_A$, $u_B$ be sampled from the groups
$U(N)$ and $U(M)$ of unitary matrices according to their Haar measures. Then the following relation holds,
\beq
\label{eq:epUVavgTheo}
\begin{split}
& \left \br e_p\left[ U (u_A\otimes u_B) V \right]\right \kt_{u_A, u_B} \\ &=
e_p(U)+e_p(V) -\, e_p(U)e_p(V)/\overline{e_p},
\end{split}
\eeq
where   $\overline{e_p}=\br e_p(W)\kt_W$ 
denotes the mean entangling power
averaged over random unitary matrices $W$ 
sampled according to the Haar measure on 
$U(NM)$.
\end{theorem}

\begin{corollary}
Let $ U^{(n)} \equiv U\,(u_{A_{n-1}}\otimes u_{B_{n-1}})\,U\ldots (u_{A_{1}}\otimes
u_{B_{1}})\, U$, where $u_{{A}_j} \in U(N)$ and $u_{B_j}\in U(M)$ are unitary matrices.
Let $V=U^{(n-1)}$, so that $U^{(n)}=U (u_{A_{n-1}} \otimes u_{B_{n-1}}) V$, then from the
theorem above
\begin{align}\label{eq:epUavg_nloc_Theo}
\begin{split}
\br e_p{(U^{(n)})}\kt_{\rm Loc} &= e_p(U) + \left[1- \frac{e_p(U)}{\overline{e_p}}\right]\br
e_p(U^{(n-1)})\kt_{\rm Loc} 
\\
&= \overline{e_p}\left[1 - \left(1 - \frac{e_p(U)}{\overline{e_p}}\right)^n\right],
\end{split}
\end{align}
where $\langle \ \cdot \ \rangle_{\rm Loc}$ denotes 
averaging over the set of local operators $u_{A_{n-1}},u_{B_{n-1}} \cdots
u_{A_1},u_{B_1} $ generated independently according to the Haar measure.
\end{corollary}

 %

Similarly it follows that
\beq
\label{eq:gtUpown}
\br g_t{(U^{(n)})}\kt_{\rm Loc} = {\bar g_t} 
\left[1 - \left(1 - \frac{g_t(U)}{\overline{g_t}}\right)^n\right],
\eeq
where the average value $ {\bar g_t}$ is
provided in Eq. (\ref{eq:generalepgtavg})
 and  for $M=N$ reduces to $1/2$. 

A proof of this result is indicated in Appendix~\ref{app:details}.

Eqs.~\eqref{eq:epUVavgTheo} and~\eqref{eq:gtUpown} constitute our main results concerning thermalization
of properties of quantum gates iterated sequentially in discrete time steps.
For any bi-partite gate $U$ with arbitrary small, but positive entangling power, its repeated application
 with local unitaries sandwiched between consequtive time steps,
 leads to a  generic  gate with entangling power and gate typicality
 characteristic to the average over the ensemble of Haar random gates from $U(NM)$. The same is illustrated in Fig.~\ref{fig:k2k3} for qubits and qutrits in the $e_p-g_t$ plane. The evolution of $U^n$ and $\langle U^{(n)}\rangle_{\rm Loc}$ is shown for a particular (non-generic) choice of the initial unitary $U$, which selected from vicinity of a local gate, as $e_p(U)$ and $g_t(U)$ are sufficiently small. While $U^n$ explores the set $K_N$ in a ``billiard" like dynamics \cite{Mandarino2018}, $\langle U^{(n)}\rangle_{\rm Loc}$ converges exponentially to the CUE average.
\begin{figure}[!tbp]
 \centering
\includegraphics[width=0.5\textwidth]{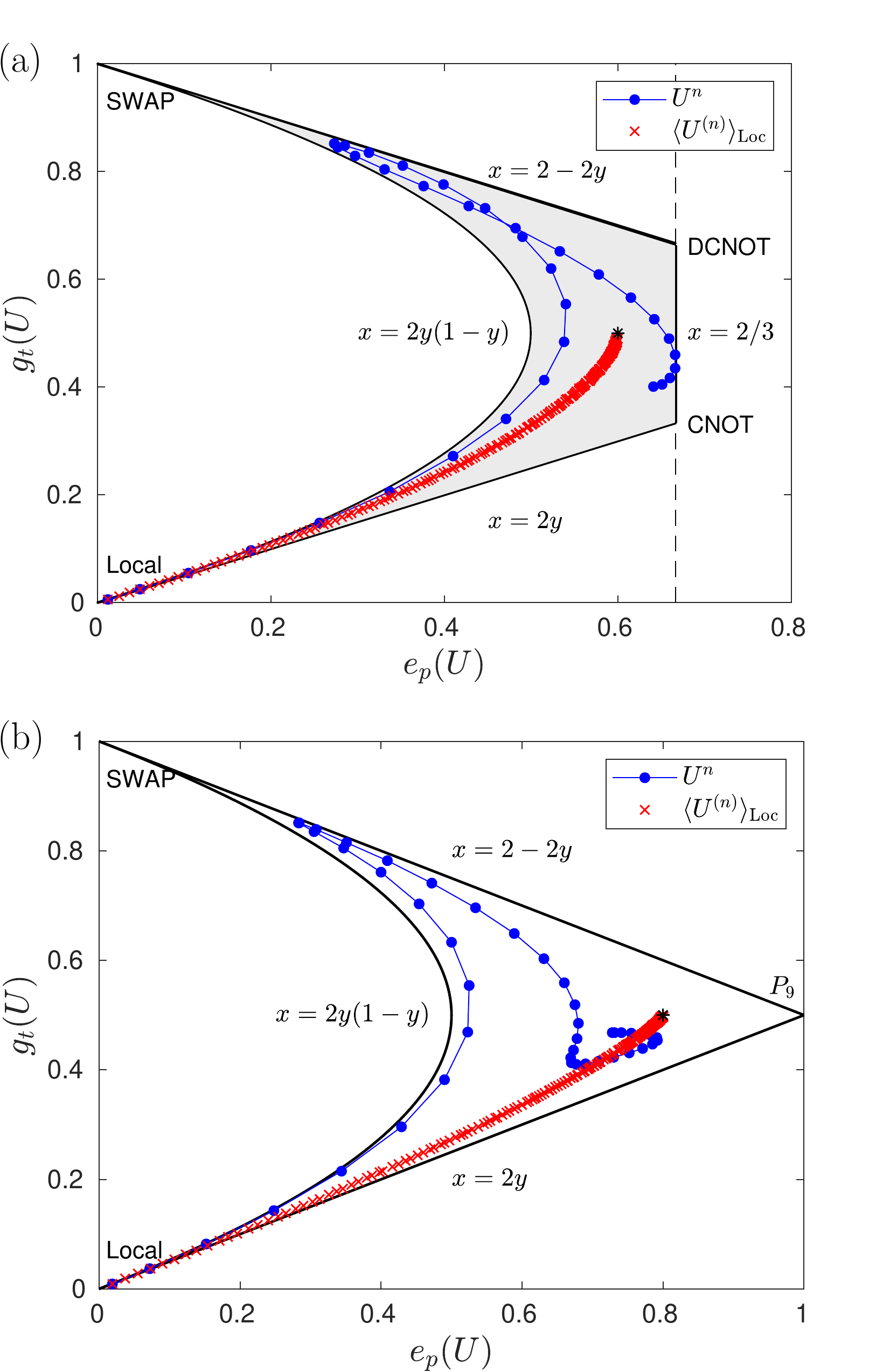}
\caption{Exemplary dynamics in the set $K_N$ 
starting in a vicinity of a local gate 
for (a) $N=2$ and (b) $N=3$.
The average value 
  $\langle U^{(n)}\rangle_{\rm Loc}$,
  represented by red crosses  ($\times$), 
  saturates to the CUE average (*),
    while $U^n$ follows a billiard type dynamics \cite{Mandarino2018}. Blue line joining the points in $U^n$ is added to guide the eye. At each time step, the average is taken over $10^4$ local random unitaries.  
}
     \label{fig:k2k3}
    \end{figure}
 
 Our results involve averaging over different local operators at
each time step and may be considered a foil for quantities such as $e_p[ ((u_A \otimes
u_B) U)^n]$ if $u_{A,B}$ are sufficiently random and have no special
relationship with $U$. Thus while 
the above results may be applicable for non-autonomous Floquet systems, 
they are also of relevance to autonomous ones.
In the case of a many-body spin chain, the  effect of thermalization of the average 
entangling power to equilibrium  has recently been reported \cite{PalLak2018}
for the symmetric case of $N=M$.
The generalization presented here allows us to extend such studies 
of thermalization to the important case of different number of spins in each subsystem.

\subsubsection{ Example: random diagonal nonlocal operators}

In \cite{JMZL2017} the entangling powers of $U^n$ and $U^{(n)}$ were evaluated for a few gates $U$, for the symmetric case $M=N$. Here we augment these results significantly,
by numerically showing that for $N \neq M$, the
thermalization of the entangling power to its average value $\overline{e_p}$ 
 holds also in the case of  a  very small interaction between both subsystems.
In particular, we analyze below the smallest interesting case of a qubit-qutrit system.


Consider a diagonal unitary matrix on $\mathcal{H}_N^A\otimes \mathcal{H}^B_M$
with entries
\beq
\label{eq:diagrand}
 (U_{\epsilon})_{m\alpha;n\beta} = e^{2\pi i\epsilon \xi_{m\alpha}}\delta_{mn}\delta_{\alpha\beta},
\eeq
 where $\epsilon\in[0,1]$ and $\xi$ is chosen
randomly and uniformly from $[-1/2,1/2)$.
Such diagonal unitaries are used to model interactions in several deterministic Floquet operators 
\cite{Shashi16,Lakshminarayan2016,Tomsovic2018}. While $\epsilon=0$ is evidently the case of zero interaction, 
$\epsilon=1$ represents the  maximal interaction.
 As $\xi$ is a random variable, $(U_{\epsilon})_{m\alpha;n\beta}$ defines an ensemble of entangling gates.
 Their entangling power was studied in \cite{Lakshminarayan2014} for
the case $\epsilon=1$, while for general $\epsilon$, it has been used in studies of
spectral transitions and entanglement \cite{Shashi16,Lakshminarayan2016,Tomsovic2018}.

\begin{figure}[!tbp]
 \centering
\includegraphics[width=0.5\textwidth]{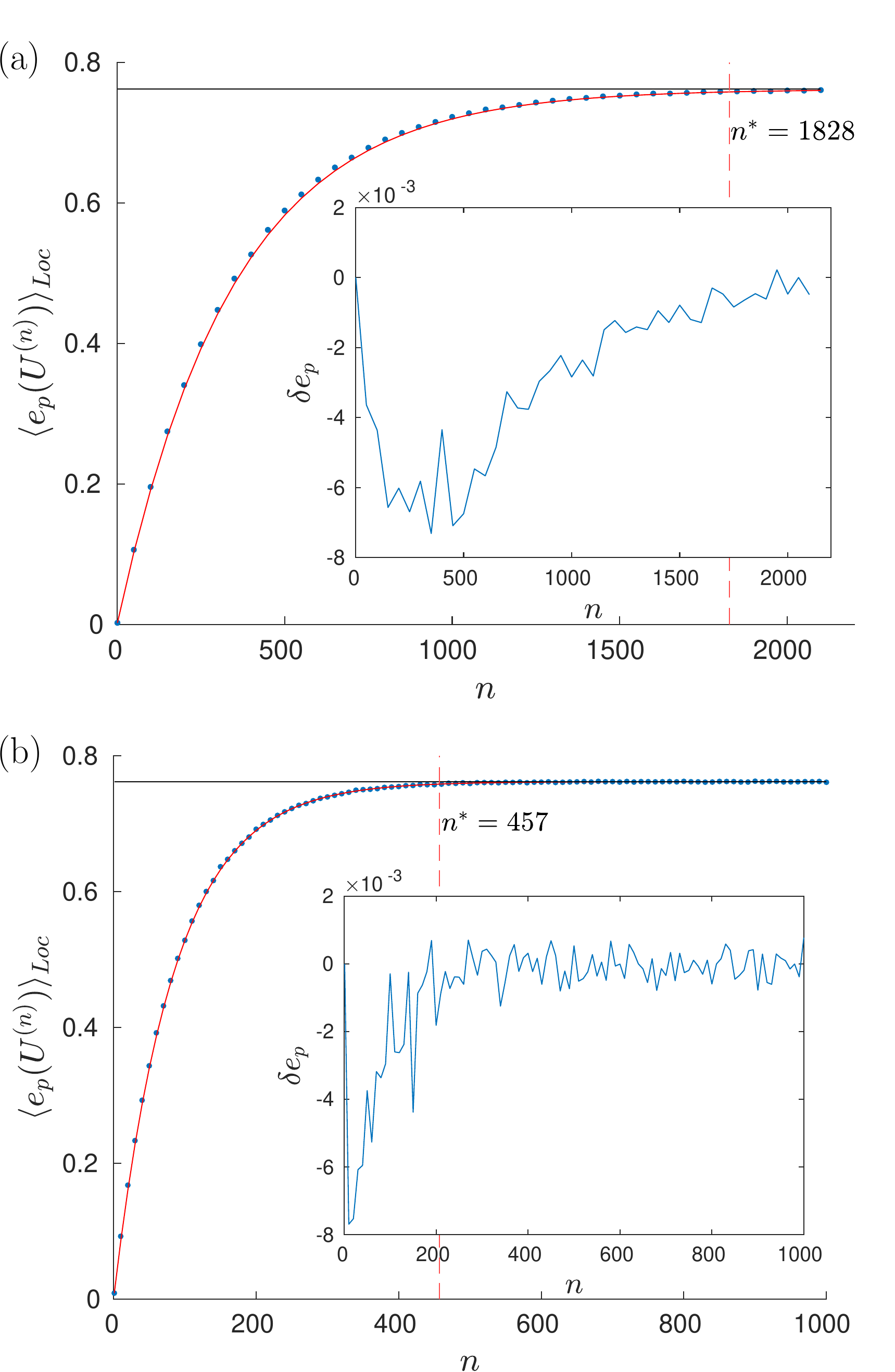}
\caption{Saturation of entangling power $e_p$ for subsystems
of size $N=2$ and $M=3$ as a function of discrete time $n$
for  two values of the interaction parameter $\epsilon$, $\epsilon = 0.025$ and 0.05 in (a) and (b), respectively.  
At each time step, the average is taken over $10^4$ local
random unitaries. 
Solid curve represents  Eq.~\eqref{eq:epUavg_nloc_Theo}, while
 the horizontal line denotes the Haar average $\overline{e_p}$  given by
Eq.~\eqref{eq:generalepgtavg}. Saturation  time  scales as $1/\epsilon^2$, see \eqref{eq:nstar}. 
The inset shows deviations from the theoretical values. }
     \label{fig:epsilon_diag}
    \end{figure}
    
For a fixed realization of the global diagonal $U_{\epsilon}$, even if $\epsilon$ is very
small, $\langle e_p(U_{\epsilon}^{(n)})\rangle_{\rm Loc}$ reaches the Haar average
$\overline{e_p}$ due to interlacing of random local unitaries as illustrated in
Fig.~(\ref{fig:epsilon_diag}) in the case $MN=6$. For $e_p(U_{\epsilon}) \ll 1$,
Eq.~\eqref{eq:epUavg_nloc_Theo} implies that
\beq
\br e_p{(U_{\epsilon}^{(n)})}\kt_{\rm Loc}\approx \overline{e_p}\, \left[1-\exp \left(-n
\frac{e_p(U_{\epsilon})}{\overline{e_p}}\right)\right].
\eeq
The saturation value reached is equal to  the global average $\overline{e_p}$
which according to Eq.~(\ref{eq:generalepgtavg}) reads $16/21 \approx 0.762$ for $N=2$ and $M=3$.
Smaller the interaction parameter, the longer it takes to 
thermalize and reach the asymptotic value.
 Deviations from the theoretical curve shown in the insets of Fig.~(\ref{fig:epsilon_diag}) are of the order of $1/\sqrt{n_{loc}}$, where $n_{loc}$ denotes the number of realization of local gates over which the averaging is done.
  Hence the number of locals $n^*$ required to push $e_p(U_{\epsilon}^{(n)})$ to the Haar average depends on $\epsilon$ as $n^* \sim\overline{e_p}/e_p(U_{\epsilon})$.

The time of thermalization can be estimated for the case of local evolution
given by the tensor product of diagonal random gates.
For a diagonal unitary $U_{\epsilon}$ of size $NM\times NM$, the reshuffled matrix is of
size $N^2\times M^2$ with $N(N-1)$ rows and $M(M-1)$ columns equal to zero. To
compute $\tr(U^R(U^R)^\dagger)^2$ in Eq.~\eqref{eq:gen_ep}, it is  thus sufficient to consider
$\tr(AA^\dagger)^2$, where $A$ is obtained by reshaping the diagonal of $U_{\epsilon}$;
$A_{jk} = e^{2\pi i\epsilon\xi_{\alpha}}$, $\alpha = (j-1)M+k$, $j=1,\dots,N$,
$k=1,\dots,M$. Here $AA^\dagger$ is a Hermitian matrix of size $N$ and for $\epsilon$ small,
$(AA^\dagger)_{jj} = M$ and off-diagonal entries $(AA^\dagger)_{jk}\approx M(1\pm
i\epsilon)$, $j<k$. Thus,
\beq\label{eq:X_diag_NM}
\tr(U_{\epsilon}^R(U_{\epsilon}^R)^\dagger)^2 = \tr(AA^\dagger)^2 \approx N^2M^2 - N(N-1)M^2 \epsilon^2
\eeq
The partial transpose of a diagonal unitary remains unchanged, hence
\beq\label{eq:Y_diag_NM}
\tr(U_{\epsilon}^{T_A}(U_{\epsilon}^{T_A})^\dagger)^2 = NM.
\eeq
Inserting Eq.~\eqref{eq:X_diag_NM} and Eq.~\eqref{eq:Y_diag_NM} in Eq.~\eqref{eq:gen_ep}
gives $e_p(U_\epsilon)\approx N\epsilon^2/(N+1)$, and therefore
\beq
\label{eq:nstar}
n^*\approx \frac{(N+1)(M^2-1)}{M(NM+1)}\; \frac{1}{\epsilon^2},
\eeq
which for $N=2$, $M=3$ gives $n^* \approx 8 \epsilon^{-2} /7$.
For $\epsilon=0.025$ and $\epsilon=0.05$, the numerical values read $n^*\approx 1828$ and $n^* \approx 457$ respectively, as shown in Fig.~(\ref{fig:epsilon_diag}). 

For large dimensions $M$ and $N$, 
one may average additionally over the diagonal ensemble of the entangling gates themselves.
It is possible to approximate such an ensemble averaged $\tr(AA^\dagger)^2$ by $N^2M^2\text{sinc}^4
(\pi \epsilon ) \approx (1-2 \pi^2 \epsilon^2/3)N^2M^2$,
where  sinc$(x) = \sin x/ x$. Hence $e_p(U_{\epsilon})\approx
2 \pi^2 \epsilon^2/3$ and $n^*\sim 3 \overline{e_p}/(2 \pi^2 \epsilon^2)$, 
so that  the saturation time  scales  as $1/\epsilon^2$.

Time evolution of quantum entanglement for initially separable states has been the subject of many studies \cite{MillerSarkar1999,BandyoLak2002, Fujisaki2003,BandyoLak2004, Dobrzanski2004,PetitJean2006}, 
often in the context of weakly interacting highly chaotic systems.
A recent study \cite{Jethin2020} combining a recursive application of perturbation theory and the theory of 
random matrices indicates an exponential saturation of entanglement measures and is consistent with our findings.
 The approach advocated here  is not perturbative
and it is based on averaging of the entangling power
 over independent local operators at each time step.
 The rate at which the average $\langle e_p(U^{(n)})\rangle $ approaches the global RMT value $\bar e_p$
  depends only on the entangling power of the nonlocal single-step operator $U$ and 
  is hence fully interaction driven.

Note that the techniques applied in this work are not 
sensitive to the degree of chaos in the classical model consisting of two uncoupled systems.
Thus analyzing the time evolution of averaged entangling power 
we are not in position to investigate the role of the Lyapunov exponent
of the corresponding classical system,
which was found essential \cite{PetitJean2006} 
for the rate of growth of the average entanglement of quantum states initially localized in the phase-space. 
That the entangling power averages over all initial product states equally, implies that any 
special properties that arise for coherent initial states are washed out.
However, further work is needed for clarifying the connections and differences between both approaches.


\subsection{Thermalization of the spectra of reshuffled and partial transposed unitaries}
We have analyzed above, how  the local unitary invariants of entangling power and operator entanglement,
and equivalently,  the entropies of the density matrices in \eqref{eq:rhoRandTdefn},
thermalize in time to their asymptotic values. However, this only reflects a more detailed approach to 
equilibrium of the spectra of related operators. In particular, it is illuminating 
to analyze complex eigenvalues of non-unitary
reshuffled and partially transposed matrices, $U^{(n) R}$ and $U^{(n) T_A}$,
which allow us to infer, to what extent the analyzed gate
approches properties characteristic to generic unitary matrices.

A large non-hermitian random matrix $G$ from the Ginibre ensemble,
 containing independent  complex random Gaussian entries,
 displays spectrum covering uniformly the unit disk, 
according to the universal circular law of Girko \cite{Girko1985,Bai1997}.
If a random unitary matrix $U$ is large enough,
the unitarity constraints become  so weak that after reshuffling
the matrix $U^R$ shows statistical properties  close to these of the Ginibre ensemble
\cite{Kus2013,Mandarino2018} - see also recent rigorous results~\cite{MPS20}.
Thus the corresponding  positive matrix, $U^R(U^R)^{\dagger}$,
 display spectra in agreement with the 
  the Mar{\v c}enko-Pastur law~\cite{MPlaw},
   $P_{MP}(x)=(2 \pi)^{-1}\sqrt{(4-x)/x}$,
  derived to describe the spectral density of   random Wishart matrices $GG^{\dagger}$.

 Let $x_i$ denote  eigenvalues of the density matrices $\rho_R(U^{(n)})$ or $\rho_T(U^{(n)})$
 rescaled by the dimension $N^2$,
 which are equal to scaled squared singular values of $(U^{(n)})^R$ and $(U^{(n)})^{T_A}$
 respectively.
 The thermalization of properties of  the gate $U^{(n)}$ with the time $n$ 
 will be reflected in the distribution  $P(x)$, 
 which for a large dimension $N$ converges to 
  the distribution $P_{MP}(x)$.
We will introduce  local averaged purities of both auxiliary density matrices,
\beq
\label{eqn:XnYn}
X_n = \br \tr[\rho_R^2(U^{(n)}]\kt ,\qquad Y_n = \br \tr[\rho_T^2(U^{(n)})]\kt .
\eeq
Then  Mar{\v c}enko-Pastur law implies that $X_n$ and $Y_n$ are of the order of $2/N^2$. A  recursion relations for these quantities starting from $(X_1,Y_1)$  was derived in \cite{JMZL2017} for the symmetric case, $M=N$. We will now demonstrate  thermalization in the spectra of density operators $\rho_R$ and $\rho_T$
for the model of the diagonal unitary ensemble and  controlled unitaries.

\subsubsection{Spectral properties for random diagonal nonlocal operators}
Consider the special case of the model 
with nonlocal matrix $U$ being  diagonal with random phases, as in Eq.~(\ref{eq:diagrand}).
To focus on the effect of time evolution itself, we set the interaction strength to
the maximal value,  $\epsilon=1$, choose  $N=M$, and  denote the diagonal nonlocal matrix by $U=U_d$. 
Average purities of the density matrices 
defined in  Eq.~(\ref{eqn:XnYn})  for $n=1$ read, 
\beq
\overline{X_1} = \frac{2N-1}{N^2} \quad \text{and} \quad \overline{Y_1} = \frac{1}{N^2}.
\eeq
In this case there are no local operators and the averaging indicates only the average with 
respect to the random phases of the nonlocal operators $U_d$.
The first one is easy to derive from the reshuffled operator, see \cite{Lakshminarayan2014}
and since the partial transpose of a diagonal unitary matrix remains diagonal,
$\rho_{T_A}(U_ d)=I_{N^2}/N^2$ hence, $Y_1 = 1/N^2$.
Thus typical diagonal unitaries, even for $\epsilon=1$, are far from being 
thermalized, although their entangling power $e_p(U_d)= (N-1)/(N+1)$ is large. This
follows from Eq.~(\ref{eq:ep_symm}), see also \cite{Lakshminarayan2014},
in  which a different normalization of the entangling power is used.

 \begin{figure}[!tbp]
   \centering
\includegraphics[width=0.55\textwidth]{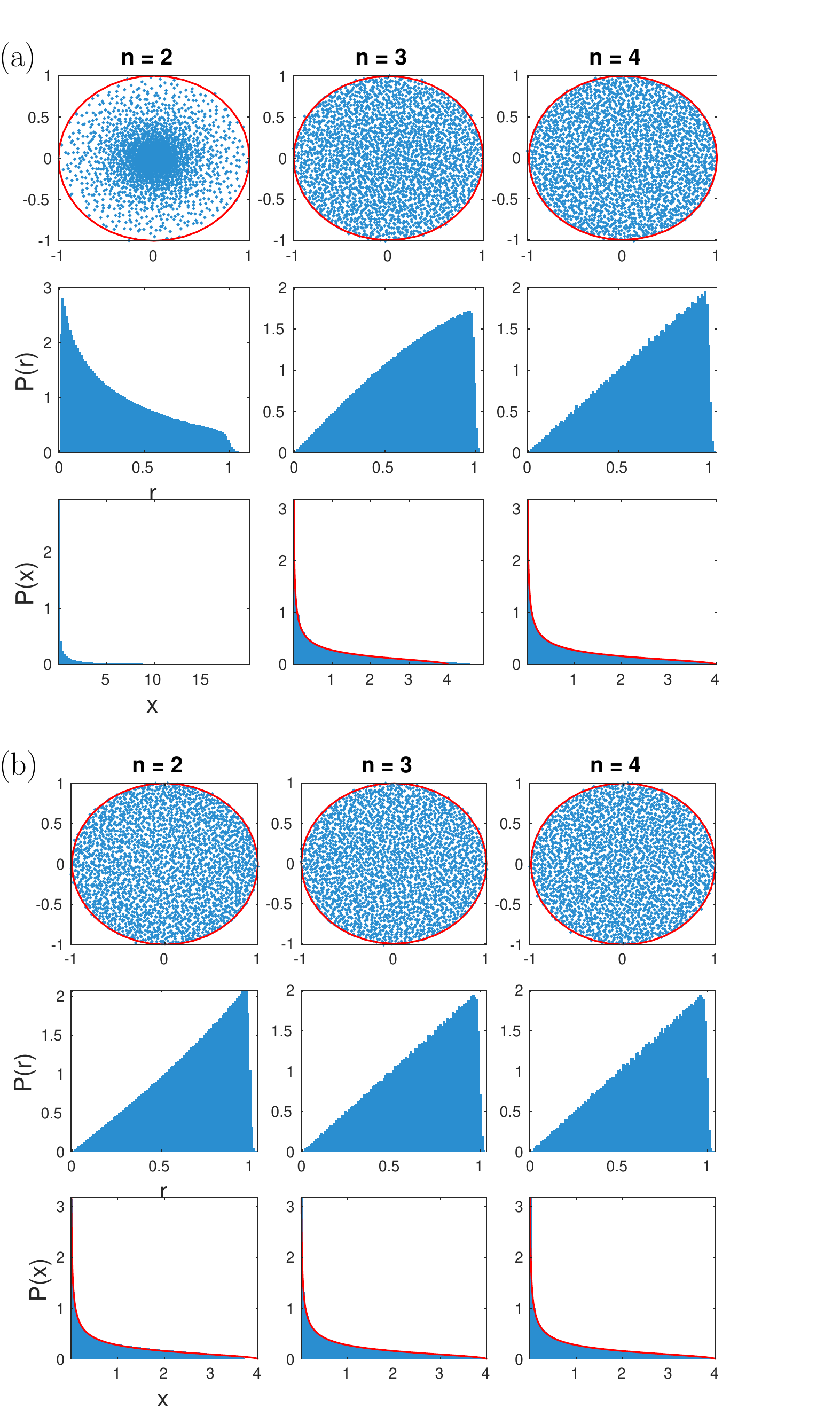}
    \caption{Spectral thermalization of $U^{(n)}$ 
    for a random diagonal nonlocal unitary matrix $U_d$ of dimension $N^2$ 
    with $N=50$. Distributions corresponding to (a) of reshuffled matrix 
     $(U^{(n)})^R$ and (b) the partially transposed matrix $(U^{(n)})^{T_A}$,
     are shown for times $n=2,3$ and $4$.
      Top rows show the complex eigenvalues  $z$,
       while  middle rows show their radial density, $P(r)$ with $r=|z|$. 
      The bottom
      rows show the distribution of
       scaled squared singular values of $(U^{(n)})^R$ and $(U^{(n)})^{T_A}$ respectively,
       which are  compared with the  Mar{\v c}enko-Pastur distribution (solid curve).}
     \label{fig:dist_ud}
    \end{figure}

For $n=2$ we consider an interlacing dynamics 
determined by random local unitary operators acting between two nonlocal operators,
$U^{(2)}=U_d (u_1\otimes u_2) U_d$, and obtain
\beq
\overline{X_2} = \frac{6}{N^2+1} \quad \text{and} \quad \overline{Y_2} = \frac{2
(N^4+N^2+1)}{N^4(N+1)^2}.
\eeq
Since $\overline{Y_2} \sim 2/N^2$, this quantity
related to  the partial transpose of $U_d$,
 is close to its asymptotic value
already after two applications of 
typical nonlocal diagonal operators. On the other hand,
the dual quantity  $\overline{X_2}$ behaves as  $6/N^2$,
which  indicates significant deviations from typicality. 

These effects are visible in Fig.~(\ref{fig:dist_ud}) in multiple ways.
 The eigenvalues of 
$(U_d^{(2)})^R$ are not distributed uniformly inside  the unit disk, which is the case for the spectrum of  $(U_d^{(2)})^{T_A}$. 
 For the former operator there are several small eigenvalues 
which reflects the fact  at $n=1$, the matrix $U_d^R$ is of rank $N$, rather than $N^2$. 
Even in the case of the partial transpose, there are visible deviations from linear structure of 
the radial distribution, which are not observed  for $n \ge 3$.
Although $\overline{X_3} \sim 2/N^2$, there exist deviations in the radial
distribution, which thus serve as a sensitive indicator of thermalization. At $n=4$
 the properties of the partial transpose and the reshuffled matrix
 are close to the matrices from the Ginibre ensemble 
 of dimension $N^2$, and the singular values follow the Mar{\v c}enko-Pastur law to a 
good approximation.

\subsubsection{Spectral properties for controlled unitary operators}

While the diagonal nonlocal operator lead to fast thermalization,
for some other models this process occurs considerably slower.
 Consider a controlled unitary operator acting on a symmetric product space,
\beq
U = P_{A_1}\otimes \mathbb{1}_B + P_{A_2}\otimes u_B
\eeq
where $P_{A_i}$ are orthogonal projectors such that $P_{A_1}+P_{A_2} =\mathbb{1}_N$, $P_{A_i} P_{A_j}=\delta_{ij}P_{A_i}$,
and $u_B \in U(N)$. It is known  \cite{Cohen2013}
that any two-qubit unitary gate of Schmidt rank two 
forms a controlled-unitary
of this kind and it can be implemented with
a maximally entangled state of two qubits
and  local operations and classical
communication (LOCC).
Thus this example may be considered the simplest entangling unitary. 
The reshuffled operator reads
\beq
U^R=|P_{A_1}^R\kt \br \mathbb{1}^R|+ |P_{A_2}^R\kt \br u_B^{*R}|,
\eeq
where $|M\kt$ reshapes or vectorizes the operator with elements $M_{ij}=\br i |M|j \kt$ into a column vector with entries $M_{ij}$. Noting that
$\br M_1|M_2\kt=\tr(M_1^{\dagger} M_2)$, we get 
\beq
\begin{split}
\rho_R=&\frac{1}{N^2}U^R U^{R \dagger}=\frac{1}{N^2}\left[ N(|P_{A_1}^R \kt \br
P_{A_1}^R|+|P_{A_2}^R \kt \br P_{A_2}^R|)+ \right.\\
&\left. \tr(u_B)^* |P_{A_1}^R \kt \br P_{A_2}^R|+\tr(u_B)|P_{A_2}^R \kt \br P_{A_1}^R|
\right],
\end{split}
\eeq
which is only a rank-2 operator. In contrast, as $(u_A \otimes u_B)^{T_{A}}=u_A^T \otimes
u_B$ and as transposes of projectors
remain projectors, $U^{T_{A}}$ is also unitary and hence $\rho_T=\mathbb{1}_{N^2}/N^2$, a
maximally mixed state. These observations immediately imply that
\beq
X_1 = \frac{1}{2}+\frac{1}{2N^2}|\tr(u_B)|^2, \quad \text{and} \quad Y_1 = \frac{1}{N^2} .
\eeq

 \begin{figure}[!tbp]
   \centering
\includegraphics[width=0.55\textwidth]{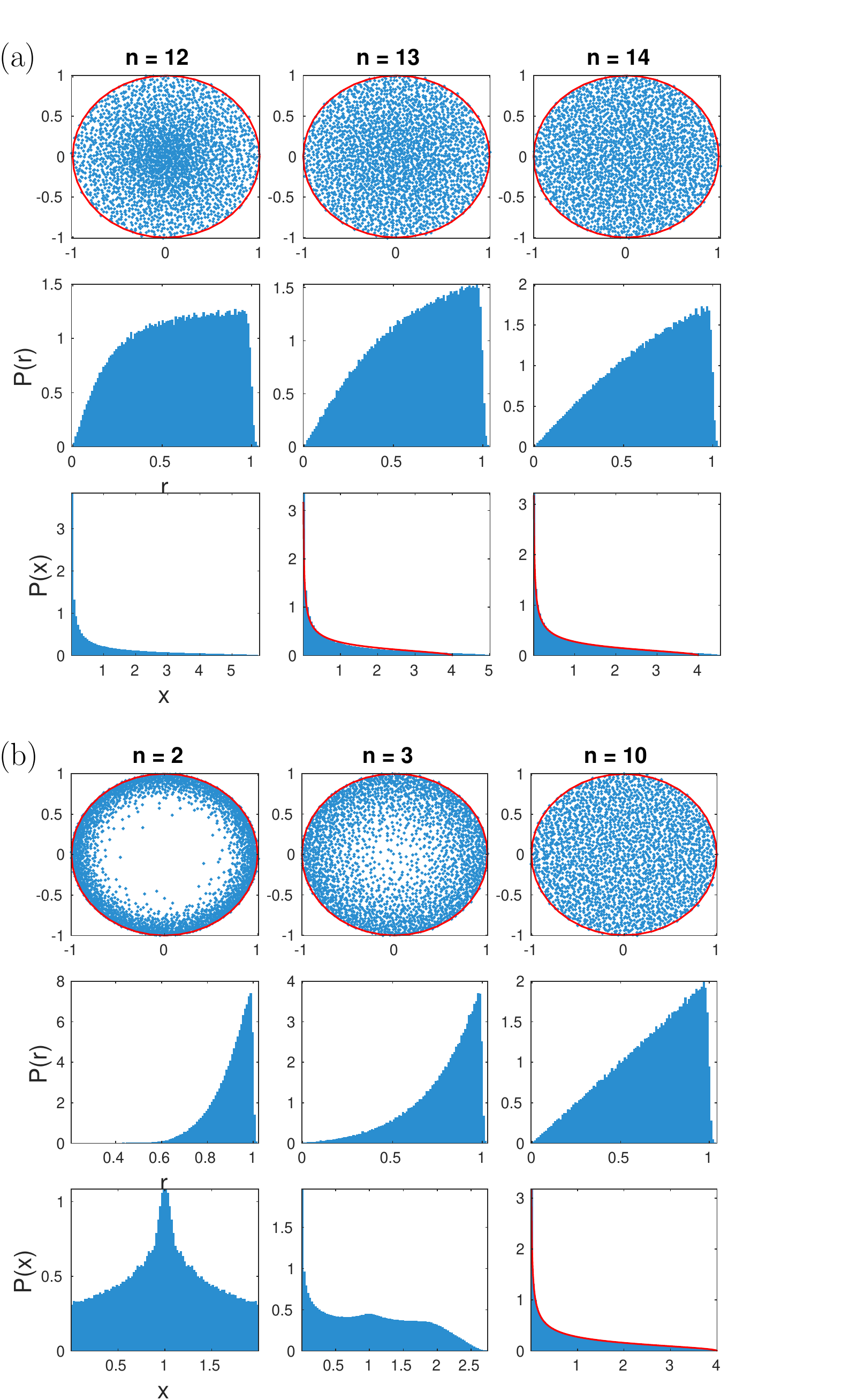}
\caption{As in Fig.  \ref{fig:dist_ud}   for $U$ selected as a controlled unitary gate. 
Third column plotted for $n =14$ and $n=10$ 
   illustrates  longer time required in this case to reach typicality.
   }
       \label{fig:dist_cu}
    \end{figure}

Taking Haar random unitary matrices  $u_B$ of size $N$  one defines  an ensemble
of controlled-unitary gates of order $N^2$, for which we evaluate
 the average purities of the associated density matrices 
 $\rho_R(U^{(n)})$ and $\rho_T(U^{(n)})$.
  Form  factor averaged over CUE matrices of size $N$
reads:
$\overline{|\tr(u_B^n)|^{2}}=n$ if $n \leq N$ and $N$ for $n \geq N$
-- see \cite{HKSSZ96}. Denoting this
additional averaging
by an overbar one obtains $\overline{X_1}\sim 1/2$
and $\overline{Y_1}= 1/N^2$.
Using the recursion relation from \cite{JMZL2017} and the CUE form factors
quoted above  for $n=2$ we arrive at,
\begin{align}
\begin{split}
\overline{X_2}&=\frac{N^6+2N^4-6N^2+4}{4N^2(N^2-1)^2} , \\
\overline{Y_2}&=\frac{5N^4 - 10N^2 +6}{4N^2(N^2-1)^2}.
\end{split}
\end{align}
The other details of the unitary gate $u^B$, are relevant to higher orders. It is also clear that
the sequence  $Y_n$
approaches the typical behavior earlier than the $X_n$. For instance, at $n=2$ we have 
$\overline{X_2}\sim 1/4$, while $\overline{Y_2}\sim  \frac{5}{4N^2}$.
In general, $ \overline{X_n} \sim \frac{1}{2^n}$, indicating that it takes a time $n^*
\sim 4 \log_2 N$ for the operators to thermalize, such that the
average operator entanglement is comparable to the average over  $U(N^2)$. 
Numerical data obtained for a typical controlled unitary gate 
presented in Figure~\ref{fig:dist_cu}  show that in this case
the thermalization time is longer in comparison to random diagonal gates. 
Even at $n=12$ one can see substantial deviations from the Girko circular law 
for the eigenvalues of the reshuffled matrix, while at $n=14$  the data 
become well described by the universal  Mar{\v c}enko-Pastur distribution. 
Spectral properties of the partially transposed matrix also reach typical behaviour around
  $n \sim 10$. These time scales are consistent with the $~\log_2N$ scale and the thermalization of the controlled-unitary gates occurs more slowly, but surely.

\section{Summary and outlook}    
\label{sec:summary}

In this work we have investigated nonlocal properties
of bipartite quantum gates acting on an $N \times M$ system.
Representing them in the plane
spanned by entangling power $e_p$ and gate typicality $g_t$,
we have analyzed the boundary of the allowed set $K_2$, 
which in turn enabled us to identify gates that correspond to critical points of the boundary 
and are distinguished by some particular properties.
Making use of the Cartan decomposition and the canonical form of a 
two-qubit gate \cite{KBG01,KC01}
we have described the boundaries analytically,
as they correspond to the edges and diagonals of the Weyl chamber.

As the Cartan decomposition is not effective for unitary matrices of order nine,
in the case of two-qutrit gates such an approach does not work, 
hence only some parts of the boundary of the set $K_3$  are known exactly.
For instance, the structure of $K_3$ is still unknown in the vicinity
of the right most point representing optimal gates, 
for which entangling power admits its maximal value,
$e_p=1$, and corresponds to maximally entangled states of a four--qutrit system. It is worth emphasizing that while such a gate does not exist for $N=2$ \cite{Higuchi2000},
the case of $N=6$ still remains open \cite{OpenProblems}.

A key issue addressed in this paper concerns  
nonlocal properties of a bipartite unitary gate applied sequentially.
Although local unitary operations performed after a single usage of a nonlocal gate
cannot change its entangling power, they do play a crucial role if  
the gate analyzed is performed several times.
Our result shows that an arbitrary small, but positive,
 entangling power of a nonlocal gate $U_{AB}$ is 
  sufficient to assure that the gate $U'_{AB}=U_{AB}V_{\rm loc}$ applied $n$ times 
will reach the entangling power typical to random unitary matrices exponentially fast.
Here $V_{\rm loc}= u_A\otimes u_B$ denotes  a random local unitary,
which is  drawn independently at each time step.
This statement illustrates the {\sl thermalization} of non-local properties of
bipartite gates with the interaction time
and sheds more light into the
properties of quantized chaotic dynamics,
in which nonlocal kicks coupling both subsystems
are interlaced by chaotic local evolution  \cite{Haake4}.

While the entangling power of a bipartite gate $U_{AB}$ determines 
the entangling power of time-evolutions augmented with local operators,
it is interesting to note that it can possibly determine 
the complexity of the corresponding many-body systems
 built out of them in various architectures.
As a concrete example, a product of $\otimes^L U_{AB}$ on a one-dimensional lattice of
$2L$ sites and its translation by one-site was studied recently
in terms of its correlation functions \cite{BKP2019}. 
It is not hard to infer from their results that the case when $U_{AB}$ has maximal 
entangling power allowed by the dimensions, corresponds to the case of a maximally chaotic many-body system. It is interesting to observe that qubits do not satisfy this condition,
while this is the case for qutrits \cite{Clarisse2005}. 
Thus we believe that our study is also relevant 
to a large body of  recent work around understanding
of quantum chaos for many-body systems.

\begin{acknowledgments}
We would like to thank Wojciech Bruzda and Dardo Goyeneche for several fruitful discussions and 
Suhail Ahmad Rather for crucial help with the qutrit phase-space. 
K.{\.Z}. acknowledges financial support by Narodowe Centrum Nauki under the
grant number DEC-2015/18/A/ST2/00274. AL and PM acknowledge financial support by the Department of Science and Technology, Govt. of India, under grant number DST/ICPS/QuST/Theme-3/2019/Q69.
\end{acknowledgments}

\appendix
\section{Bipartite unitary gates and four-party entangled pure
states}\label{app:4party}
The bijection between states on ${\cal H}^{(1)}\otimes {\cal H}^{(2)}$ and operators on
${\cal H}^{(1)} \cong {\cal H}^{(2)}$ is known in the physics 
litarature under the name Choi-Jamio{\l}kowski isomorphism,
which relates the set of pure states of a bipartite system
and the set of operations acting  on a simple system \cite{Bengtsson2007}.
Any normalized bipartite pure state $|\psi\rangle = \sum_{ij=1}^N x_{ij} |ij\rangle$ can
be written
as $(X\otimes \mathbb{1}) |\phi^+\rangle$, where 
$|\phi^+\rangle =\sum_{j=1}^N |jj\rangle /\sqrt{N}$
 is the maximally entangled state and 
$\langle i| X|j\rangle  =\, \sqrt{N} x_{ij}$.
Note that a state $|\psi\rangle$ is maximally entangled, 
if and only if the matrix $X$ is unitary,
as then its partial trace is maximally mixed, 
$\rho_A={\rm tr}_B|\psi\rangle \langle \psi|=XX^{\dagger}/N= \mathbb{1}/N$.
It is often convenient to make use of this
relation between the set $U(N)/U(1)$ of unitary quantum gates of order $N$
and the set of maximally entangled states in $N \times N$ system 
\cite{Bengtsson2007}.

The same relation can also be used in a more general set-up, if the system
${\cal H}^{(1)}$ is composite and describes two subsystems of sizes $N$ and $M\geq N$,
denoted $A$ and $B$ respectively. 
The system ${\cal H}^{(2)}$ of the same size $NM$ is also composite 
and contains two subsystems  $C$ and $D$ of dimensions $N$ and $M$ respectively. 
The matrix $X/\sqrt{NM}$ with elements $x_{i\alpha ,k \beta}$ describes now a $4$-party
pure state
$|\psi_{ABCD}\rangle = \sum_{ij }^N \sum_{\alpha \beta}^M x_{i\alpha ,k \beta } |i\alpha
k \beta \rangle$,
and can be considered as a four-index tensor
or a $NM \times NM$ matrix with composite indices.

Any bi-partite matrix $X$ acting on subsystems $AB$, thus defines a four-partite pure
state,
\begin{equation}
|\psi_{ABCD}\rangle = (X_{AB} \otimes \mathbb{1}_{CD})
|\phi_{AC}^+\rangle \otimes |\phi_{BD}^+\rangle .
\label{eq:four_parties}
\end{equation}
Note that the above formula does not factorize,
as the symbol $\otimes$ denotes tensor products
acting with respect to different partitions.
If the bipartite matrix acting on the subsystems $AB$
is unitary,  $X=U$,
 then the corresponding four party state 
$|\psi_{ABCD}\rangle$ is maximally 
entangled with respect to the partition $AB|CD$, so all
the components of the corresponding Schmidt vector of length $NM$, eigenvalues of
$\rho_{AB}={\rm tr}_{CD} |\psi_{ABCD}\rangle \langle \psi_{ABCD}|=UU^{\dagger}/NM$,
are equal to $1/NM$. Unitarity condition, $UU^{\dagger}=\mathbb{1}$,
implies then the maximal entanglement 
of the state  $|\psi_{ABCD}\rangle$
with respect to  the splitting $AB|CD$.

On the other hand, one can investigate whether
this state is entangled with respect to two 
other possible partitions, $AC|BD$ and $AD|BC$.
To this end one studies the partially reduced states
$\rho_{AC}={\rm tr}_{BD} |\psi_{ABCD}\rangle \langle \psi_{ABCD}|$ 
with spectrum $\lambda_i$, with $1 \leq i \leq N^2$
and
$\rho_{AD}={\rm tr}_{BC} |\psi_{ABCD}\rangle \langle \psi_{ABCD}|$
with spectrum $\mu_j$ with $1 \leq j \leq NM$.

For any four-index  matrix  $X_{ij,\alpha \beta}$ of size $NM$
it will be convenient to use
the following operations on its entries \cite{Bengtsson2007}:
the partial transpose, $X^{T_A}$, where $X_{i \beta,j \alpha}^{T_A}=X_{i \alpha ,j \beta
}$,
is also an $NM \times NM$ matrix, and 
the reshuffling, $X^{R}$, where $X_{ij,\alpha \beta }^{R}=X_{i \alpha ,j \beta }$ is an
$N^2 \times M^2$ dimensional array.
The first one, $T_A$,  represents transposition on the first subsystem only
and preservs hermiticity of $X$. 
The reshuffling $R$, corresponds to reshaping each block of a matrix into a vector,
does not preserve unitarity nor hermiticity.

%
It is easy to check \cite{Zyczkowski2004} that the  vector $\lambda$, 
equal to the spectrum of the positive matrix 
\beq
\label{eq:rhoR}
\rho_R(U)\equiv \rho_{AC}=U^R (U^R)^{\dagger}/(NM),
\eeq
coincides with the vector defining the  operator Schmidt decomposition 
of the scaled matrix $U$. Correspondingly, the vector $\mu$, forming the spectrum of 
\beq
\label{eq:rhoT}
\rho_{T_A}(U)\equiv \rho_{AD}=U^{T_A} (U^{T_A})^{\dagger}/(NM),
\eeq
appears in the Schmidt decomposition of the operator $U$ composed with the \textsc{swap} $S$ for
the symmetric case.

The reduced state of the subsystems $AC$ is maximally mixed 
 if  $\rho_{AC}=\mathbb{1}_{N^2}/N^2$
and corresponds to the maximum entanglement in the $AC|BD$ split. This in turn happens
when the rearrangement $U^R$ satisfies $U^R (U^R)^{\dagger}=\mathbb{1}_{N^2}/ (MN)$,
where $\mathbb{1}_{N^2}$
is the identity matrix of dimension $N^2$. 
In other words, for the symmetric case, $N=M$, if $U^R$ is also unitary,
 then $\rho_R(U)$ is maximally mixed and the subsystem
$AC$ is maximally mixed with $BD$. 
Hence the linear entanglement entropy $1-\tr_{AC}\rho_{AC}^2$,
based on the reshuffling of $U$ can serve as a measure of the entanglement in the four-party state in
Eq.~\eqref{eq:four_parties} with respect to the partition $AC|BD$.  

The state $|\psi_{ABCD}\rangle$
is maximally entangled with respect to  the third splitting $AD|BC$
if the Schmidt vector is flat, $\mu_j=1/MN$ for $j=1,\dots, NM$,
so that the matrix $X^{T_A}$ is unitary.
Entanglement for this partition can be thus characterized  
by the twin quantity $1-\tr_{AC}\rho_{AD}^2$.


Analyzing a bipartite unitary gate $U$, 
described by a four-index matrix $U_{a,b} = X_{i\alpha,j \beta}$, 
it is convenient to introduce the notion of {\sl multiunitarity} \cite{Goyeneche2015}.
For the symmetric case ($N=M$), a matrix $X$ of size $N^2$, 
wrtitten in the four-index notation, is called {\sl
$2$--unitary}, if the following three conditions are satisfied
\begin{eqnarray}
\label{eq:multi}
&
{\rm i)} \ \  X \in U(N^2) \Leftrightarrow
\sum_{ij} X_{ij,\alpha \beta }  {\bar X_{ij,\alpha' \beta'}}
=   \delta_{\alpha,\alpha'}  \delta_{\beta,\beta'},  \nonumber \\
&
{\rm ii)} \ \  X^R \in U(N^2) \Leftrightarrow
\sum_{i,\alpha} X_{i j,\alpha \beta }  {\bar X_{ij',\alpha \beta'}}=   \delta_{j,j'}  \delta_{\beta,\beta'},
 \nonumber \\
&
{\rm iii)} \ \  X^{T_A} \in U(N^2) \Leftrightarrow
 \sum_{i,\beta} X_{i j,\alpha \beta }  {\bar X_{ij',\alpha' \beta}}=   \delta_{j,j'}  \delta_{\alpha,\alpha'},
 \nonumber
  \\
  & 
\end{eqnarray}
so apart of $U$, also
 two other matrices with interchanged entries, $U^{T_A}$
and $U^R$, are unitary.
The corresponding four-index tensor $X_{i\alpha,j \beta}$ of size $N$, 
is called {\sl perfect},
if for any choice of two indices  out of four,
the matrix  of size $N^2$ obtained 
 by restructuring the four-index tensor into a matrix is unitary \cite{PYHP15}.
By construction,  any  $2$-unitary matrix  $U$ of order $N^2$
provides an example of a matrix which maximizes the entangling power, $e_p(U)=1$,
as both linear entanglement entropies $E(U)$ and  $E(US)$ are maximal.
Thus the corresponding four-party state (\ref{eq:four_parties})
is maximally entangled with respect to all three possible partitions.
Such states are called two-uniform \cite{Sc04} or absolutely maximally entangled (AME)
\cite{HCLRL12}. 

Interestingly, such states do not exist in a four-qubit system 
\cite{Higuchi2000}, as the total size of the Hilbert space is too 
small to find a state satisfying all necessary constraints.
This is equivalent to the known fact~\cite{Zanardi2000, Clarisse2005}
that there is no unitary matrix of size $N^2=4$,
for which the maximal value $e_p=1$  
of the entangling power is achieved,
which is consistent with the structure of the set $K_2$ 
plotted in Fig.  \ref{fig:ep_vs_gt}.
 In the complementary notation, there are no $2$-unitary
matrices of order four~\cite{Goyeneche2015}.
On the other hand AME states exists for larger systems
consisting of four qutrits,
 which is equivalent to the statement that there exists a $2$-unitary 
matrix of size $N^2=9$, which maximizes the entangling power $e_p$ \cite{Clarisse2005}.
For any $N=3,4,5$ and $N\ge 7$ there exist permutations matrices of size $N^2$ which are
$2$-unitary, and hence maximize the entangling power \cite{Clarisse2005}
and also correspond to AME states of four systems with $N$ levels each.
For $N=6$,  the non-existence of any $2$-unitary permutation
matrix of order $36$ is directly related to the famous 
problem of $36$ officers by Euler
and follows from the 
non-existence of two mutually orthogonal Latin Squares of size six.
The more general question as to whether there exists a $2$-unitary matrix
of size $N^2=36$ (not necessarily a permutation)
remains open \cite{OpenProblems}.

Two unitarity of a bipartite gate $U$,
 corresponds to a two-uniform pure state of four parties,
 maximally entangled with respect to all three partitions.
 Sometimes it is interesting to relax one requirement
 and analyze pure state  $|\psi_{ABCD}\rangle$
  for which only two partial traces out of three  are maximally mixed.
 This weaker condition corresponds to a unitary matrix $U$ of size $N^2$
 such that additionally $U^{T_A}$ {\sl or} $U^R$ is unitary.
 The class of unitary matrices such that the partial transposition $U^{T_A}$ 
 remains unitary  was studied 
in context of quantum operations preserving some given matrix algebra,
and a method to generate them numerically based on a kind of 
the Sinkhorn algorithm was proposed  \cite{DNP16,BN17}.
Such a technique based on alternating projections on manifolds converges \cite{LM08}, 
if we wish to assure that two unitarity conditions are satisfied,
so that two partial traces of the corresponding four-party state are fixed  \cite{DLP20},
but it will usually become less effective  if  three
conditions (\ref{eq:multi}) need to be fulfilled simultaneously.

The observations made in this Appendix for the symmetric case, $M=N$,
 can be  summarized  as follows. 
\begin{proposition}\label{prop:2u_ame}
For any unitary operator $U$ acting on a bipartite space $\mathcal{H}_{N}^{A}\otimes
\mathcal{H}_{N}^{B}$, the following are equivalent.
\begin{itemize}
\item[(a)] The unitary $U\in U(N^2)$ attains the global maximum of entangling power, that is,  $e_p(U)=1$, as both linear entanglement entropies $E(U)$ and  $E(US)$ are maximal.
\item[(b)] The bipartite unitary matrix $U$ is $2$-unitary. In other words both the transformed matrices $U^R$ and $U^{T_A}$ remain unitary.
\item[(c)] If $U_{AB}=U$, the pure state \[ |\psi_{ABCD}\rangle = (U_{AB} \otimes
\mathbb{1}_{CD})|\phi_{AC}^+\rangle \otimes |\phi_{BD}^+\rangle \] defined in
Eq.~\eqref{eq:four_parties} is maximally
entangled with respect to all possible bipartitions and thus forms an absolutely maximally entangled state of four quNits.
\item[(d)] The corresponding four-index tensor $u_{i\alpha,k\beta}$ whose elements describe the four-partite state
\[ |\psi_{ABCD}\rangle = \sum_{i,j=1 }^N \sum_{\alpha, \beta=1}^N u_{i\alpha ,j \beta }
|i\alpha j \beta \rangle\] is perfect.
\end{itemize}
\end{proposition}

\section{The stationarity of the parabola of powers of \textsc{swap}}
\label{app:lemmaswap}
\begin{lemma}\label{lem:swap1}
Define $f(u) \equiv e_{p}(u) - 2g_{t}(u) \bigl( 1-g_t(u) \bigr)$, 
where $e_{p}(u)$ and $g_{t}(u)$ are respectively the entangling power and gate typicality of a bipartite unitary $u$. The function $f(u)$ is extremised whenever $u=U_t=e^{itS}$ is a fractional power of the \textsc{swap} operator.
\end{lemma}
\begin{proof}
Operators close to arbitrary fractional powers of \textsc{swap} $U_t$, with $0 < \epsilon \ll 1$ 
are 
\beq
U_{t,\epsilon}=\exp( itS +i \epsilon H) \approx \exp(i t S)\left(\mathbb{1}+i \epsilon
H\right),
\eeq
where $H$ is a Hermitian operator. We may require without loss of generality that $H$ is
traceless, that is $\tr H=0$, as the overall phase will make no difference to
calculations. We may also assume that
$H$ is orthogonal to $S$, that is $\tr(HS)=\tr(SH)=0$, as any overlap with $S$ will be
equivalent to only shifting $t$ to a new value.
The difference $\delta U_t=U_{t,\epsilon}-U_t$ is given by 
\beq
\label{eq:deltaUt}
\delta U_t = i \epsilon \left(\cos t\,  H  + i \sin t \, SH \right).
\eeq
We will show that $\delta E(U_t)=0$ and $\delta E(U_tS)=0$, thus under such perturbations
$\delta e_p(U_t)=0$ and $\delta g_t(U_t)=0$ and finally $\delta f(u=U_t)=0$.

From Eq.~(\ref{Eq:EUUR}) it follows that 
\beq
\delta E(U)=-\frac{4}{N^2}\, \text{Re}\, \tr \left( \delta U^R U^{R\,\dagger} U^R U^{R
\,\dagger}\right).
\eeq
From the linearity of the reshuffling operation,
$\delta U_t^R=(U_{t,\epsilon})^R - U_t^R=(\delta U_t)^R$.
From this and Eq.~(\ref{eq:deltaUt}) we get 
\beq
\delta U_t^R =(\delta U_t)^R= i\epsilon \left( \cos t H^R + i \sin t (SH)^R \right).
\eeq
To show that $\tr \left( \delta U_t^R U_t^{R\,\dagger} U_t^R U_t^{R \,\dagger}\right)=0$, we note that it
involves $\tr(H^R \mathbb{1}_R)$, $\tr(H^R S)$, $\tr((SH)^R \mathbb{1}_R)$ and
$\tr((SH)^R S)$. It is straightforward to verify that when $H$ is orthogonal to $S$ and is traceless,
all of these vanish. In a similar way it is easy to show also that
$\tr \left( \delta U_t^{T_A} U_t^{T_A \, \dagger} U_t^T U_t^{T_A \,\dagger}\right)=0$.
Thus $\delta f(u=U_t) =0$ when $u=U_t$ is a power of the \textsc{swap} $S$, except when $\delta u$ is along $S$. In the latter case, $f(u)=0$ strictly and there is no variation of $f$, establishing that $f(u)$ is indeed an extremum
 if $u=e^{itS}$ is a fractional power of the \textsc{swap}.
\end{proof}

\section{Proof of the theorem concerning 
average entangling power
$e_p[\mathcal{U}^{(n)}]$}
\label{app:details}

\begin{theorem}
Let $U$ and $V$ be unitary operators on $\mathcal{H}^A_N\otimes\mathcal{H}^B_M$ and
$u_A$, $u_B$ be sampled from the groups
$U(N)$ and $U(M)$ of unitary matrices according to their Haar measures, then 
\beq
\begin{split}
& \left \br e_p\left[ U (u_A\otimes u_B) V \right]\right \kt_{u_A, u_B} \\ &=
e_p(U)+e_p(V) -\, e_p(U)e_p(V)/\overline{e_p},
\end{split}
\eeq
where the average entnagling power reads
$\overline{e_p}=\br e_p(W)\kt_{W}$,
 and $W$ is sampled according to the Haar measure
on the unitary group
$U(NM)$.
\end{theorem}

\begin{proof}
Consider an extended Hilbert space
$\mathcal{H}_N^A\otimes\mathcal{H}_M^B\otimes\mathcal{H}^C_A\otimes\mathcal{H}^D_M$ where
$\mathcal{H}^C_N$ are $\mathcal{H}^D_M$ are copies of $\mathcal{H}^A_N$ and
$\mathcal{H}^B_M$. Using the identity $\tr(\rho_A \otimes \rho_C S_{AC})=\tr( \rho_A^2)$
where $\rho_C$ is a copy of $\rho_A$ and $S_{AC}$ is the \textsc{swap} operator, the entangling power of $U$ acting on $\mathcal{H}^A_N\otimes\mathcal{H}^B_M$
was written in \citep{Zanardi2000} as
\beq
\label{eq:ep_diffdim}
\begin{split}
e_p(U) &= \frac{2}{\tilde{e}_p^{\text{max}}}\,\tr(U^{\otimes
2}\Omega^{++}_p{U^\dagger}^{\otimes 2}\Pi^-_{AC}).
\end{split}
\eeq
Here $\Pi^{-}_{AC}=2^{-1}(\mathbb{1}- S_{AC})$ is the projector over the anti-symmetric
subspace of $\mathcal{H}^{A}_N\otimes\mathcal{H}^{C}_N$,
and 
$\Omega^{++}_p=\omega_{AC}^+ \otimes \omega_{BD}^+$ and $\omega_{AC}^+=\int d\mu(\psi_A)
(|\psi_A \kt \br \psi_A|\otimes |\psi_A\kt \br \psi_A|)$,
while $\omega_{BC}^{+}$ is an identical operator. When $d\mu(\psi_A)$ is the Haar measure
on states in $\mathcal{H}^A_N$, recognizing that
$\omega^{+}_{AC}$ has support only on the symmetric subspace, group theoretic
arguments involving Schur's lemma were used in \citep{Zanardi2000} to show that
$\Omega^{++}_p = 4C_AC_B\Pi^+_{AC}\Pi^+_{BD}$. Here $C_A^{-1} =N(N+1)$, $C_B^{-1}
=M(M+1)$, $\Pi^{+}_{AC}=2^{-1}(\mathbb{1}+ S_{AC})$ is the projector over the symmetric
subspace of $\mathcal{H}^{A}_N\otimes\mathcal{H}^{C}_N$, while $\Pi^{+}_{BD}$ is a similar projector on $\mathcal{H}^{B}_M\otimes\mathcal{H}^{D}_M$.

This forms a convenient starting point for us, as the local unitary averaged entangling power is
\beq
\label{eq:ep_diffdim_2}
\left\langle e_p(U (u_A \otimes u_B) V)\right\rangle_{u_A,u_B} =
\frac{2}{\tilde{e}_p^{\text{max}}}\,\tr(U^{\otimes 2}\left\langle
Q\right\rangle{U^\dagger}^{\otimes 2}\Pi^-_{AC})
\eeq
where $Q = V\otimes V\Omega^{++}_pV^\dagger\otimes V^\dagger$, and 
\beq
\begin{split}
\left\langle Q\right\rangle = &\int d\mu(u_A)\,d\mu(u_B)(u_A\otimes u_B)^{\otimes 2}\,\\
&V^{\otimes 2}\,\Omega^{++}_p\, {V^\dagger}^{\otimes 2}\,(u_A^{\dagger} \otimes u_B^{
\dagger})^{\otimes 2}.
\end{split}
\eeq
Since the local unitaries are sampled independently, the average over $u_A$, $u_B$ can be
done separately. Note that $V\otimes V$ acts on $AB$ and its copy $CD$, while
$\Omega^{++}_p$ acts on $AC$ and $BD$ independently.
Note also that $\left\langle Q\right\rangle$ is self-adjoint and hence diagonalizable.
For any $x_A\in U(N)$, $[(x_A)^{\otimes 2},\left\langle Q\right\rangle] = 0$ due to the unitary invariance of the Haar measure. With similar reasoning $[(x_B)^{\otimes
2},\left\langle Q\right\rangle] = 0$, $\forall x_B\in U(M)$. Since $(x_A)^{\otimes 2}$,
$(x_B)^{\otimes 2}$ acts irreducibly on the totally symmetric and anti-symmetric
subspaces, it follows from the above commutation relations and Schur's lemma \cite{Cornwell1997} that $\br Q \kt$ can be written as a linear
combinations of projectors on the symmetric and anti-symmetric subspaces,
\beq
\label{eq:Q_avg}
\begin{split}
\br Q\kt = &\alpha_1\Pi^+_{AC}\Pi^+_{BD} + \alpha_2\Pi^-_{AC}\Pi^-_{BD} + \\
&\alpha_3\Pi^+_{AC}\Pi^-_{BD} + \alpha_4\Pi^-_{AC}\Pi^+_{BD}, 
\end{split}
\eeq
where $\alpha_l =
[\tr(\Pi^{\pm}_{AC}\Pi^{\pm}_{BD})]^{-1}\tr(Q\,\Pi^{\pm}_{AC}\Pi^{\pm}_{BD})$;
$l=\{1,\dots,4\}$.
That the operator $Q$ can be used for finding $\alpha_l$ instead of $\br Q\kt$ follows
from the fact that $(u_A^{ \dagger} \otimes u_C^{ \dagger}) \Pi^{+}_{AC}
( u_{A} \otimes u_{C}) = \Pi^{+}_{AC}$.

Next, we evaluate expressions for $\tr(Q)$ and $\tr(Q\, S_{AC})$, as follows 
(summation over repeated indices is assumed):
\begin{align}\label{eq:a1}
\begin{split}
&\tr(Q) = \tr\left(\Omega^{++}_p\right) = 1,\\
&\tr(Q\, S_{AC}\otimes S_{BD})\\
&= \frac{1}{N(N+1)}\frac{1}{M(M+1)}\tr\left(V\otimes V(1+S_{AC}+S_{BD}\right.\\
&\quad\qquad\qquad\qquad\qquad \left. {}+S_{AC}\otimes S_{BD})V^\dagger\otimes V^\dagger
S_{AC}\otimes S_{BD}\right)
\end{split}
\end{align}

Now,
\begin{align}\label{eq:a2}
\begin{split}
&\tr\left(V\otimes V\,S_{AC}V^\dagger\otimes V^\dagger\,S_{AC}\otimes S_{BD}\right)\\
=& \br i_1\alpha_1 j_1\beta_1\vert V\otimes V\vert i_2\alpha_2 j_2\beta_2\kt\br
j_2\alpha_2 i_2\beta_2\vert V^\dagger\otimes V^\dagger\kt j_2\beta_2 i_2\alpha_2\kt\\
 =&N^2M
\end{split} 
\end{align}
Similarly,
\begin{align}
\begin{split}
\tr\left(V\otimes V S_{BD}V^\dagger\otimes V^\dagger S_{AC}\otimes S_{BD}\right) &=
NM^2\\
\tr\left(V\otimes V S_{AC}\otimes S_{BD}V^\dagger\otimes V^\dagger S_{AC}\otimes
S_{BD}\right) &= N^2M^2
\end{split}
\end{align}
Combining these trace relations in Eq.~\eqref{eq:a1} gives,
\beq
\tr(Q\, S_{AC}\otimes S_{BD})=1
\eeq
To compute $\tr(Q\,S_{AC})$ and $\tr(Q\,S_{BD})$, note that
\begin{align}\label{eq:rhoR_intermsof_R_swaps}
\begin{split}
\tr(\rho_{AC}^2) &= \frac{1}{N^2M^2}\tr\left[\left(V^R(V^R)^\dagger\right)^2\right]\\
&=\frac{1}{N^2M^2}\tr\left(V\otimes V S_{AC}V^\dagger\otimes V^\dagger S_{AC}\right),
\end{split}
\end{align}
where the equality in the second line can be seen via a similar calculation as in
Eq.~\eqref{eq:a2} (see also \cite{Wang2003}). Similarly,
\begin{align}\label{eq:rhoT_intermsof_T_swaps}
\begin{split}
\tr(\rho_{AD}^2) &=
\frac{1}{N^2M^2}\tr\left[\left(V^{T_A}(V^{T_A})^\dagger\right)^2\right]\\
&=\frac{1}{N^2M^2}\tr\left(V\otimes V S_{BD}V^\dagger\otimes V^\dagger S_{AC}\right).
\end{split}
\end{align}
Using Eq.~\eqref{eq:gen_ep}, Eq.~\eqref{eq:rhoR_intermsof_R_swaps}, and
Eq.~\eqref{eq:rhoT_intermsof_T_swaps},
\begin{align}
\begin{split}
&\tr(Q\,S_{AC})\\
& =\frac{1}{N(N+1)}\frac{1}{M(M+1)}\tr\left(V\otimes V(1+S_{AC}+S_{BD}\right.\\
&\quad\qquad\qquad\qquad\qquad \left. {}+S_{AC}\otimes S_{BD})V^\dagger\otimes V^\dagger
S_{AC}\right)\\
&= \frac{1}{N(N+1)}\frac{1}{M(M+1)}\tr\left(NM^2 + N^2M\right.\\
&\quad\qquad\qquad\qquad \left. {}+\tr\left(V^R(V^R)^\dagger\right)^2 +
\tr\left(V^{T_A}(V^{T_A})^\dagger\right)^2\right)\\
&=1 - \tilde{e}_p^{\text{max}} e_p(V)
\end{split}
\end{align}
Similarly,
\beq
\tr(Q\,S_{AC}) = 1 - \tilde{e}_p^{\text{max}} e_p(V).
\eeq

Using Eq.~\eqref{eq:ep_diffdim_2} and the traces evaluated above, we get
\begin{align}\label{eq:traces_Q}
\begin{split}
\tr(Q) &= \tr[Q\,(S_{AC}\otimes S_{BD})] =1,\\
\tr(Q\,S_{AC}) &= \tr(Q\,S_{BD}) = 1 - \tilde{e}_p^{\text{max}} e_p(V) .
\end{split}
\end{align}

Thus the local unitary averaged entangling power is given by
\beq
\br e_p[U(u_A \otimes u_B ) V]\kt_{u_A,u_B} = e_p(U) + \left[1-
\frac{e_p(U)}{\overline{e_p}}\right]e_p(V),
\eeq
where $\overline{e_p}$ is the CUE averaged entangling power in
Eq.~(\ref{eq:generalepgtavg}).
\end{proof}

\begin{corollary}
The local unitary averaged gate typicality is \cite{JMZL2017},
\beq
\br g_t{(U^{(n)})}\kt_{\rm Loc} = {\bar g_t}\left[1 - \left(1 - g_t(U)/ {\bar g_t} \right)^n \right],
\eeq
where the average value $ {\bar g_t}$
is given in Eq. (\ref{eq:generalepgtavg}).

\end{corollary}
\begin{proof}
When $M=N$, gate typicality in Eq.~\eqref{eqn:gt} is given by
\beq
g_t(U) = \tr(U^{\otimes 2}\Omega_p^{+-}{U^{\dagger}}^{\otimes 2}\Pi_{AC}^-),
\eeq
where $\Omega_p^{+-} = C_A\tilde{C}_B\Pi^+_{AC}\Pi^-_{BD}$, $\tilde{C}^{-1}_B = N(N-1)$.
Starting with the above relation for $M\neq N$ and proceeding the same way as in the proof
of entangling power, proves the corollary.
\end{proof}

\end{document}